\documentclass{article}
\pdfoutput=1

\usepackage[preprint]{neurips_2026}

\usepackage[utf8]{inputenc} 
\usepackage[T1]{fontenc}  
\usepackage{hyperref}      
\usepackage{url}            
\usepackage{booktabs}      
\usepackage{amsfonts}       
\usepackage{nicefrac}       
\usepackage{microtype}      
\usepackage{subcaption}
\usepackage{xcolor}         
\usepackage{graphicx}
\usepackage{amsmath}
\usepackage{amsthm}
\usepackage{amssymb}

\newtheorem{proposition}{Proposition}
\newtheorem{lemma}{Lemma}

\title{FrOGS: Discrete Neural Sampler for Independent Alloy Configurations Across Chemical Conditions}

\author{%
  Kyucheol Min \quad Elyssa Hofgard \quad Tess Smidt \\
  Massachusetts Institute of Technology \\
  Cambridge, MA 02139 \\
  \texttt{\{mink202, ehofgard, tsmidt\}@mit.edu} \\
}

\begin{document}

\maketitle

\begin{abstract}
Predicting the thermodynamic properties of an alloy requires sampling its configurations across many chemical conditions and recovering free energies on a common absolute scale. Markov chain Monte Carlo (MCMC) is the standard tool, but it requires separate simulations at different conditions, and auxiliary free-energy methods such as thermodynamic integration are used to place results on a common absolute scale. Modern discrete neural samplers typically use reverse KL divergence as the objective and can be mode-seeking or biased. We present Free energy Offering Generative Sampler (FrOGS), a hybrid discrete neural sampler that couples an autoregressive model to a continuous-time Markov chain (CTMC) to be trained jointly under a single shared loss. FrOGS draws i.i.d.\ configurations, returns an unbiased estimate of the partition function, and gives consistent estimates of thermodynamic observables. We train a single model across a wide range of chemical conditions to produce estimates on a common absolute free-energy scale. FrOGS matches exact finite-size results on the 2D Ising model and reference phase diagrams for AgPd and CuAu, without mode collapse. We additionally compare to SEGAL, a published autoregressive baseline, and find that only FrOGS recovers the stability range of the CuAu$_3$ phase.
\end{abstract}

\section{Introduction}

 Sampling the different ways constituent atoms can occupy sites in an alloy according to their energetics is a core task in computational materials science, enabling the prediction of thermodynamic behavior and material properties. The established framework for predicting the properties of an alloy consists of two steps: (1) fitting a computationally inexpensive cluster expansion (CE) to computationally expensive first-principles data such as density functional theory (DFT) calculations~\citep{Sanchez1984, Connolly1983}; (2) sampling with Markov chain Monte Carlo (MCMC) in the semi-grand canonical (SGC) ensemble on the CE to obtain thermodynamic averages. Integrating the averages over the chemical condition $(\Delta \mu, T)$~\citep{vandeWalle2002mc} gives free energies and phase diagrams, where $\Delta \mu$ is the difference in chemical potential between two atomic species and $T$ is the temperature.

MCMC is a well-established method, yet it has known drawbacks. A Markov chain is built to sample one specific target distribution at a given $(\Delta \mu, T)$, and its updates are inherently sequential. Thus, to obtain the phase diagram across chemical conditions, one must do separate MCMC runs. The free energies are then obtained only through an auxiliary calculation that places the separate runs on a common absolute scale~\citep{vandeWalle2002mc}. Additionally, near phase boundaries, MCMC may require more steps to produce reasonable configurations~\citep{berg1991multicanonical, wolff1989collective}. 

Neural samplers have been proposed to address some of the drawbacks of MCMC. However, they can miss entire modes of the target distribution without increasing their loss, a failure known as mode collapse. This is a known consequence of training using the reverse KL divergence~\citep{minka2005divergence, soletskyi2025mode}. To address these limitations, we introduce the Free energy Offering Generative Sampler (FrOGS), which builds upon the Locally Equivariant discrete Annealed Proactive Sampler (LEAPS) framework~\citep{holderrieth2025leaps}.

In \S\ref{sec:related}, we review previous work on neural sampling and complementary MCMC approaches. In \S\ref{sec:LEAPS}, we provide background on the LEAPS framework and present our model in \S\ref{sec:method}. In \S\ref{sec:experiments}, we confirm that FrOGS reproduces the exact finite-size thermodynamics of the 2D Ising model and reconstructs the phase diagrams of AgPd and CuAu. We additionally compare to a prior autoregressive model named SEGAL~\citep{damewood2022sampling} and find that FrOGS more accurately reconstructs the phase diagrams for AgPd and CuAu.

\begin{figure}
  \centering
   \includegraphics[width=\linewidth]{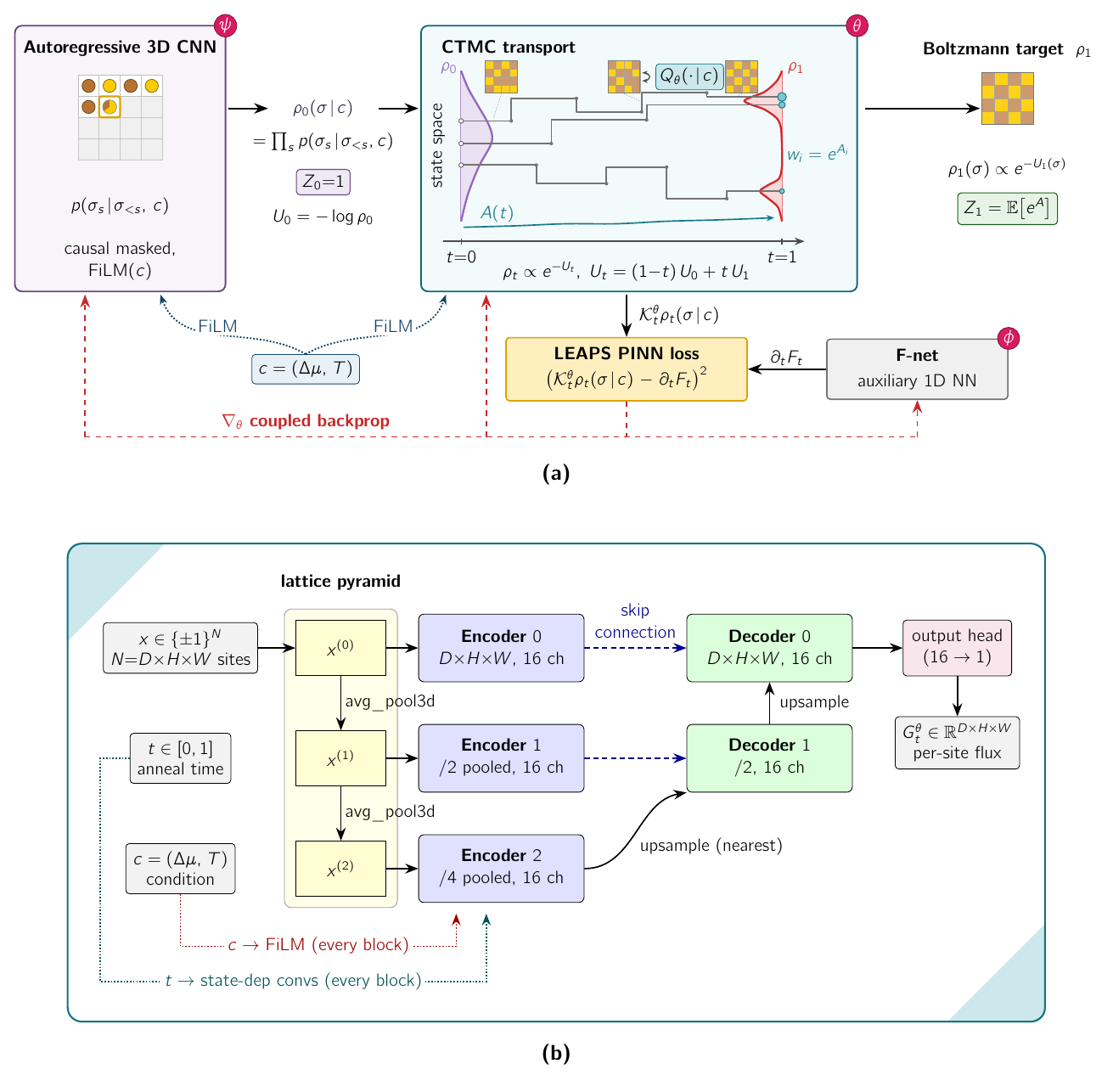}
  \caption{(a) FrOGS model architecture. FrOGS couples the continuous-time Markov chain (CTMC) transport (blue), which is based on the LEAPS algorithm, to an autoregressive network (purple) to flexibly model the prior. The cells with a red circle on the upper right side mark the neural-network-parametrized parts of the framework. (b) Architecture of the rate matrix network for the CTMC transport. The encoder and decoder blocks are described in the main text (\S\ref{sec:arch}).}
  \label{fig:conv_unet_film}
\end{figure}

Our main contributions are:

\begin{itemize}
    \item \textbf{One trained sampler for the entire $c=(\Delta\mu, T)$ region.} A single trained FrOGS model samples microstates across chemical conditions $c$ on a common absolute scale. Feature-wise Linear Modulation (FiLM) conditioning on $c$, which
    provably preserves local equivariance, enables a single model to generate samples at each chemical condition.

    \item \textbf{Hybrid architecture.} FrOGS takes a novel approach of coupling an autoregressive network as a prior to the neural CTMC transport via a shared loss. The prior absorbs the discontinuity at a first-order boundary, enabling the representation of multiple phases with one model, as discussed in \S\ref{sec:phasemix}. The downstream CTMC part is the learned generalization of annealed importance sampling~\citep{Neal2001AIS, holderrieth2025leaps}, thereby refining the proposal the prior makes.

    \item \textbf{Effective sampling in the multimodal regime without mode collapse.} By comparing the free energy at fixed composition (\S\ref{sec:phasediagram}), we recover the miscibility gap of AgPd and the regions of all three ordered phases of CuAu, attesting to the accuracy of the free-energy prediction even near the phase boundaries.

  \item \textbf{Extending the LEAPS framework.} We generalize the LEAPS framework to be applicable to real chemical systems. A locally equivariant multi-scale head lets the network capture long-range interactions, and the reparametrization of the rate matrix avoids the numerically unstable density ratio computation at low temperatures.

\end{itemize}

\section{Related Work}
\label{sec:related}

Autoregressive models are appealing for lattice systems because they
provide exact, normalized likelihoods, making free energies and thermodynamic
observables directly computable~\citep{wu2019van, nicoli2020asymptotically, damewood2022sampling}. Recent work has improved their
scalability through any-order training, in
which the model learns to generate sites in a random order, and marginalization~\citep{du2026scaling}. In these autoregressive models, the underlying
objective minimizes a reverse KL divergence and is therefore mode-seeking, so the
learned distribution can fail to cover coexisting phases near
transitions~\citep{minka2005divergence, soletskyi2025mode}. Furthermore, their biases are usually corrected by an unlearned reweighting scheme such as importance sampling~\citep{nicoli2020asymptotically}. Our model improves on both axes, since our upstream autoregressive prior does not take the reverse KL divergence as the loss function, and the remaining bias is corrected by importance sampling combined with learned CTMC.

An alternative approach keeps Monte Carlo sampling but reshapes its dynamics to overcome the barriers that slow plain MCMC near phase transitions. Metadynamics~\citep{Laio2002Metadynamics} adds a history-dependent bias alongside a collective variable (a low-dimensional summary of the configuration) so that the system is pushed out of minima it has already visited. Its
``well-tempered'' variant~\citep{Barducci2008WellTempered} adds a smaller bias as the bias accumulates in a region of the phase diagram, so that the total bias converges. The recovered profile of metadynamics is a projection onto the collective variable, so states that the collective variable cannot distinguish can be grouped into a single basin~\citep{Bussi2020Metadynamics}. \citet{du2026metadns} extends well-tempered metadynamics into
discrete neural samplers, flattening the landscape during training, but each model is trained at a fixed thermodynamic condition and recovers a free-energy profile up to an additive constant.

Flat-histogram methods~\citep{berg1991multicanonical, WangLandau2001} take a related approach, directly estimating the density of states in order to obtain a uniform distribution of composition and energy. This delivers the free energy globally on an absolute scale, but the histogram uniformity is imposed only on the
marginal of composition and energy. Atomic configurations can share the same composition and energy, so this method does not balance sampling across these other degrees of freedom~\citep{neuhaus2003crystal, dayal2004performance}. Furthermore, the method does not natively provide independent configurations, as the algorithm only uses macroscopic variables.

A separate line of work learns the transport itself, rather than extending MCMC-based methods. Neural networks are trained to move probability mass from a tractable prior to a complex target. \citet{Vaikuntanathan2008Escorted} first showed that augmenting
annealing dynamics~\citep{Neal2001AIS} with a transport field can suppress the variance, while
noting that the field cannot in general be obtained analytically. With the rise of machine learning, learned samplers emerged to approximate the solutions to the transport field. Normalizing flows learn transport
maps toward Boltzmann targets for equilibrium sampling and free-energy
estimation~\citep{Noe2019Boltzmann, Albergo2019Flow, Wirnsberger2020Targeted}.
Annealed transport samplers couple a learned flow at each step with a sequential Monte Carlo sampler in order to remain asymptotically unbiased~\citep{Arbel2021AFT, Matthews2022CRAFT}. The learned flow is trained by a per-transition KL objective, so mode coverage is supplied by the resampling and MCMC steps. A separate line casts sampling as stochastic
optimal control of a diffusion model, learning the control for generating the
target~\citep{ZhangChen2022PIS, Vargas2023DDS, RichterBerner2024}. NETS
\citep{albergo2025nets} keeps the prescribed annealing path but learns a single continuous-time drift trained to minimize the
variance of the weights. LEAPS~\citep{holderrieth2025leaps} carries this construction to
discrete state spaces, which we apply to lattice alloys.

\section{Background}
\label{sec:LEAPS}
 
We build on LEAPS~\citep{holderrieth2025leaps}, introduced in \S\ref{sec:related} within the context of learned transport. Let $\sigma=(\sigma_1,\dots,\sigma_N)$
denote a configuration on a lattice of $N$ sites, where $\sigma_i$ is the species
occupying site $i$ (\S\ref{sec:sgc}). In the framework, a CTMC $(X_t)_{t\in[0,1]}$ with learnable parameters $\theta$ evolves a
configuration $X_t$ under a rate matrix $Q^\theta_t$, transporting mass along the
annealing path $\rho_t=\tfrac{1}{Z_t}e^{-U_t}$ from a tractable prior $\rho_0$
(with $X_0\sim\rho_0$) to the target $\rho_1=p_c$, the Boltzmann distribution at a
chemical condition $c=(\Delta\mu, T)$. $Z_t$ is the normalization constant at each $t$, also known as the partition function. The rate matrix is restricted to
single-site moves: $Q^\theta_t(\tau,i\mid\sigma)\ge 0$ is the rate of the move
$\sigma\mapsto\sigma'=\mathrm{Swap}(\sigma,i,\tau)$ that replaces the species at
site $i$ by $\tau\neq\sigma_i$. Additionally, LEAPS employs importance sampling: at each $t$, each of the $M$ samples is reweighted by a log-weight $A_t \in \mathbb{R}$. The ESS (effective sample size) is defined with the log-weights as
\begin{equation}
    \mathrm{ESS}
    = \frac{\big(\sum_{m=1}^{M} e^{A^{(m)}}\big)^2}
           {M\sum_{m=1}^{M} e^{2A^{(m)}}}\ \in \Big(\tfrac{1}{M},1\Big],
    \label{eq:ess}
\end{equation}
which equals $1$ when all weights are equal (perfect transport) and approaches
$1/M$ when one sample dominates. An estimator built from $M\cdot\mathrm{ESS}$
independent samples from $\rho_t$ would have about the same variance as one built from the $M$ weighted samples~\citep{KongLiuWong1994}. We provide further detail of this framework in Appendix~\ref{sec:leapsth}.

LEAPS offers two main benefits: (a) a direct unbiased estimate of $Z_1$ with the equation 
\begin{equation}
    \mathbb{E}[e^{A_t}]=Z_t/Z_0
    \label{eq:partition_t}
\end{equation}
 and (b) a consistent estimate of any observable through reweighting the samples, as
 
 \begin{equation}
\langle O\rangle = \frac{\mathbb{E}\!\left[e^{A_t}\,O(X_t)\right]}{\mathbb{E}\!\left[e^{A_t}\right]}.
\label{eq:reweighting}
 \end{equation}

LEAPS, however, faces challenges when directly applied to the alloy sampling task. Firstly, as \citet{du2026scaling} point out, the original LEAPS model is not conditionable. For example, the LEAPS model in the original paper was trained and evaluated just at the critical temperature~\citep{holderrieth2025leaps}. Secondly, LEAPS starting from a uniform prior fails to transport efficiently across first-order phase boundaries, as shown in the next section. Finally, LEAPS has only been applied to 2D systems on a square grid. The alloy systems we aim to model are 3D, exchange-symmetry-broken, and non-simple-cubic (e.g. fcc). We extend the framework to overcome the stated limitations.

\section{Method}
\label{sec:method}

\subsection{Learned prior}
\label{sec:phasemix}

We first outline how the learned prior allows FrOGS to sample across first-order phase boundaries. As shown in \S\ref{sec:LEAPS}, the LEAPS estimator is unbiased for any rate matrix $Q^\theta_t$. Thus, when sampling fails at a first-order phase boundary, it is a result of the variance of the log-weights $A_t$.

Consider the case where the target $\rho_1$ is ordered but $\rho_0$ is disordered. For example, in ordered Cu$_3$Au (L1$_2$), Au atoms sit at the corners of
the cubic unit cell and Cu atoms at the face centers, whereas in disordered
Cu$_3$Au the atoms occupy the lattice sites randomly. The annealing path crosses a first-order phase transition at some $t^* \in (0,1)$. Let the set of disordered configurations be
$\mathcal{D}$ and ordered configurations $\mathcal{O}$, and let $\rho_t(\mathcal{O})=\sum_{\sigma\in\mathcal{O}}\rho_t(\sigma)$. In a first-order phase transition near $t^*$, the probability mass will jump from approximately 0 to approximately 1 over a narrow range of $t$. As LEAPS only updates one site at a time, the transport must pass through energetically unstable configurations $\mathcal{B}$ between $\mathcal{D}$ and $\mathcal{O}$. As these configurations are unstable, $\rho_t(\mathcal{B})$ will be exponentially suppressed. Consequently, the transported distribution does not match $\rho_t(\mathcal{O})$ near $t^*$, resulting in high variance of the log-weights. The resulting heavy-tailed weights cause the ESS to collapse, as seen in Eq.~\eqref{eq:ess}.

We observe the ESS collapse in the ordered-disordered system of CuAu with the uniform prior, and we confirm this is not the result of mode collapse (\S\ref{sec:ablation}). To resolve this issue, we replace the fixed prior with a conditioned
autoregressive network $\rho_0(\sigma\mid c)=\prod_i p(\sigma_i\mid\sigma_{<i},c)$ that places mass in the target basin for each $c=(\Delta\mu, T)$, so the local transport does not have to carry the probability mass across first-order boundaries.

We use a PixelCNN-style~\citep{vandenoord2016pixelrnn} causal-masked autoregressive 3D-convolutional network (Fig. \ref{fig:auto}) re-indexed onto the FCC lattice (\S\ref{app:reindex}). Because the equations \eqref{eq:partition_t} and \eqref{eq:reweighting} hold regardless of the prior setting, we can use any prior as long as $Z_0$ is tractable. The tractability is necessary for calculating the target $Z_1$ with Eq.~\eqref{eq:partition_t}. Since it is autoregressive, $Z_0=1$ is satisfied. 

\subsection{Coupling and conditioning the autoregressive and transport networks}
\label{sec:coupling}

The learned prior allows us to jointly train the prior and the transport networks, so that one coupled model can generate samples in the entire $c=(\Delta\mu, T)$ region. We show in \S\ref{sec:ablation} that the performance degrades substantially after removing the trained prior (initializing from a uniform distribution as is done in LEAPS), and after removing the transport (training the prior only with $Q_t^\theta=0$).

The autoregressive and transport networks are coupled by a shared loss when training. A single backward pass through the objective therefore returns gradients for both networks at once (\S\ref{sec:gradient}). Importantly, the advantage of FrOGS is that its loss function is not mode-seeking since the training objective for the autoregressive prior is a generalized form of VarGrad \citep{richter2020vargrad} (\S\ref{sec:vargrad}).

We want the prior to propose all the degenerate modes of the distribution at a given chemical condition because the downstream transport cannot move probability masses across degenerate basins, similar to the phenomenon described in \S~\ref{sec:phasemix} for the first-order phase transition. To achieve this, we add a term $\lambda\,\mathcal{L}_{\mathrm{sym}}$ to the loss that explicitly penalizes the prior for dropping a degenerate mode, with $\lambda=0.1$ and
\begin{equation}
  \mathcal{L}_{\mathrm{sym}}
  =\mathbb{E}_{(\sigma,c),\ \gamma\sim\mathrm{Unif}(\mathcal{S})}
   \Big[\Big(\log\frac{\rho^{\psi}_0(\sigma\mid c)}{\rho^{\psi}_0(\gamma\sigma\mid\gamma c)}\Big)^{2}\Big],
  \label{eq:sympen}
\end{equation}
where $\mathcal{S}$ is the symmetry group of the Hamiltonian. The term equals zero at $\rho^{\psi}_0=p_c$, and it cannot bias Eqs.~\eqref{eq:partition_t} and~\eqref{eq:reweighting} because it only influences the prior (\S\ref{sec:vargrad}). 

Both the transport network and the prior network are conditioned on $c=(\Delta\mu, T)$ through FiLM \citep{Perez2018FiLM}.  We explain our detailed setup in \S\ref{sec:film}. FiLM is a general-purpose conditioning method, originally demonstrated on conditioning an image classifier on natural language for visual reasoning. It is per-channel and spatially uniform, the output at
a given site still depends only on the sites preceding it. Moreover, FiLM retains the local equivariance of the transport network, as we show in detail in \S\ref{sec:proof}. 

\subsection{Extending the transport network}
\label{sec:arch}
  
We extend the LEAPS framework in two novel directions. LEAPS models $Q^\theta_t$, the learned rate matrix, as
\begin{gather}
    Q^\theta_t(\tau,i\mid\sigma)=\big[\,G^\theta_t(\tau,i\mid\sigma)\,\big]_+, \quad  \text{where}
    \  [z]_+ \equiv \max(z,0) \ \label{eq:original} \\
    G^\theta_t(\tau,i\mid\sigma)
    = -\,G^\theta_t(\sigma_i,i\mid\sigma'), \quad \text{where } \sigma'=\mathrm{Swap}(\sigma,i,\tau) 
    \label{eq:le}
\end{gather}
for computational efficiency. Eq.~\eqref{eq:le} is referred to as the local equivariance property. The motivation behind this parametrization can be found in \S\ref{sec:leapsth} or the LEAPS paper~\citep[\S 8]{holderrieth2025leaps}.

As in LEAPS, we fulfill the local equivariance condition by factoring the probability flux as
  \begin{equation}
    G^\theta_t(\tau,i\mid\sigma)
    = \big(P^\theta_t(\tau)-P^\theta_t(\sigma_i)\big)^{\!\top} H^\theta_t(i\mid\sigma),
    \label{eq:locequiv}
  \end{equation}
into a token projector $P^\theta_t$ and a locally invariant head $H^\theta_t$ whose output at site $i$ does not depend on the current token $\sigma_i$. Swapping
  $\sigma_i\leftrightarrow\tau$ leaves $H^\theta_t$ untouched and only flips the sign
  of the projector difference, satisfying Eq.~\eqref{eq:le}.

First, we introduce a U-net-style locally equivariant architecture (Fig. \ref{fig:conv_unet_film}) for $H_t^\theta$ that consists of a pyramid of encoders and decoders. Encoder blocks follow the locally equivariant convolutional network construction of LEAPS, with FiLM inserted after each convolution. Decoder (refine) blocks are identical to encoder blocks except at the input: instead of a single lattice-derived feature, they take the concatenation of the upsampled coarse feature and the encoder feature through the skip connection, mixed back to 16 channels by a $1\times1\times1$ convolution before the first FiLM.

This architecture allows for a larger receptive field than stacking locally equivariant convolution layers and shows better performance at an equal scale (\S\ref{sec:training}). This is an example of a more scalable locally equivariant architecture that the original LEAPS paper anticipated~\citep[\S 12]{holderrieth2025leaps}. We prove that it is locally equivariant in \S\ref{sec:proof}.

Second, we reparametrize the rate matrix for the alloy systems. The objective of LEAPS requires the computation of the ratio $\rho_t(\sigma')/\rho_t(\sigma)=e^{-\Delta_{i\tau}U_t}$ between neighboring configurations (i.e. differing by one site). This ratio is then multiplied by $[-G^\theta_t]_+$ before entering the loss function. The challenge is that a large energy difference (e.g. for alloy systems at low temperatures) between the neighboring configurations ($\Delta_{i\tau}U_t\ll 0$)
  makes this ratio diverge.

We address this issue by the reparametrization below with the equilibrium Metropolis acceptance factor $M_t$.
\begin{equation}
  \begin{split}
    \tilde Q^\theta_t(\tau,i\mid\sigma)
    &= \big[G^\theta_t(\tau,i\mid\sigma)\big]_+\,M_t(\tau,i\mid\sigma), \\
    M_t(\tau,i\mid\sigma)
    &\equiv e^{-[\Delta_{i\tau}U_t]_+}
     = \min\!\big(1,\,e^{-\Delta_{i\tau}U_t}\big)\in(0,1].
  \end{split}
  \label{eq:reparam}
\end{equation}
After substituting Eq.~\eqref{eq:reparam} into the loss function of LEAPS and using the elementary identity $[a]_+-[-a]_+=a$ together with the detailed balance relation, we get a bounded, single-pass expression in which the density ratio
  $e^{-\Delta_{i\tau}U_t}$ does not appear (full derivation in \S\ref{sec:derivation}). 

  As $M_t>0$ is a
  $\theta$-independent rescaling of each move, the reparametrized
  rates~\eqref{eq:reparam} span the same family of CTMCs as $[G^\theta_t]_+$, so the
  universal-representation property~\citep[Prop.~8.1]{holderrieth2025leaps} is preserved.

\section{Experiments}
\label{sec:experiments}

We report the results of our sampler in the main text for the $L=20$ 2D Ising model, 125-site AgPd, and 128-site CuAu systems, matching the largest cells reported by the discrete neural samplers we benchmark against~\citep{damewood2022sampling, du2026scaling}, and compare them with the previous works in detail in \S\ref{sec:comp}.

All of the reported results below are from a FrOGS model trained for 200k steps for each chemical system. We injected 40 local MCMC steps per time step during training, which is a feature of the LEAPS algorithm, to ensure better coverage. We select this number based on the hyperparameter sweeps (\S\ref{sec:training}). We did not inject MCMC steps at each step in $t$ during evaluation to verify the power of our sampler alone. We construct the phase diagram using the method outlined in \S\ref{sec:phasediagram}.

For AgPd and CuAu, we sample from the SGC ensemble (\S\ref{sec:sgc}) under the cluster expansion parameters from \citet{damewood2022sampling}. We train on the same CE data from \citet{damewood2022sampling} and \citet{du2026scaling} in order to fairly compare our phase diagrams.

We remark that the thermodynamics under a particular CE can be considerably different from the real phase diagram. As an example, for CuAu, the CE reproduces the $\text{L1}_2/\text{L1}_0/\text{L1}_2$ orderings but shifts their transition temperatures (e.g. Cu$_3$Au above 700 K versus 663 K experimentally~\citep{okamoto1987aucu}). All results below should therefore be read as the thermodynamics of the given CE.

\subsection{Ising Model}
\label{sec:ising}

In Fig.~\ref{fig:ising}, we plot the statistics of the $20 \times 20$ two-dimensional Ising model with periodic boundary conditions with 5000 samples at each condition from FrOGS with unreparametrized rates of Eq.~\eqref{eq:original} trained in $[-0.05, 0.05]$ and $[1.5, 3.0]$ for $\Delta \mu$ and $T$, respectively. The free energies of the 2D Ising model are known exactly at finite size~\citep{Kaufman1949}, allowing us to compare our model to true values. The thermodynamic results from FrOGS samples match the Kaufman ground truth with high accuracy (Fig.~\ref{fig:ising}). Sampling with both the prior and the learned transport (orange in the figure) tracks the references much closer than the output of the autoregressive part with CTMC transport turned off during evaluation ($Q^\theta_t=0$, blue), showcasing the ability of the learned transport. 
\begin{figure}
    \centering
    \includegraphics[width=\linewidth]{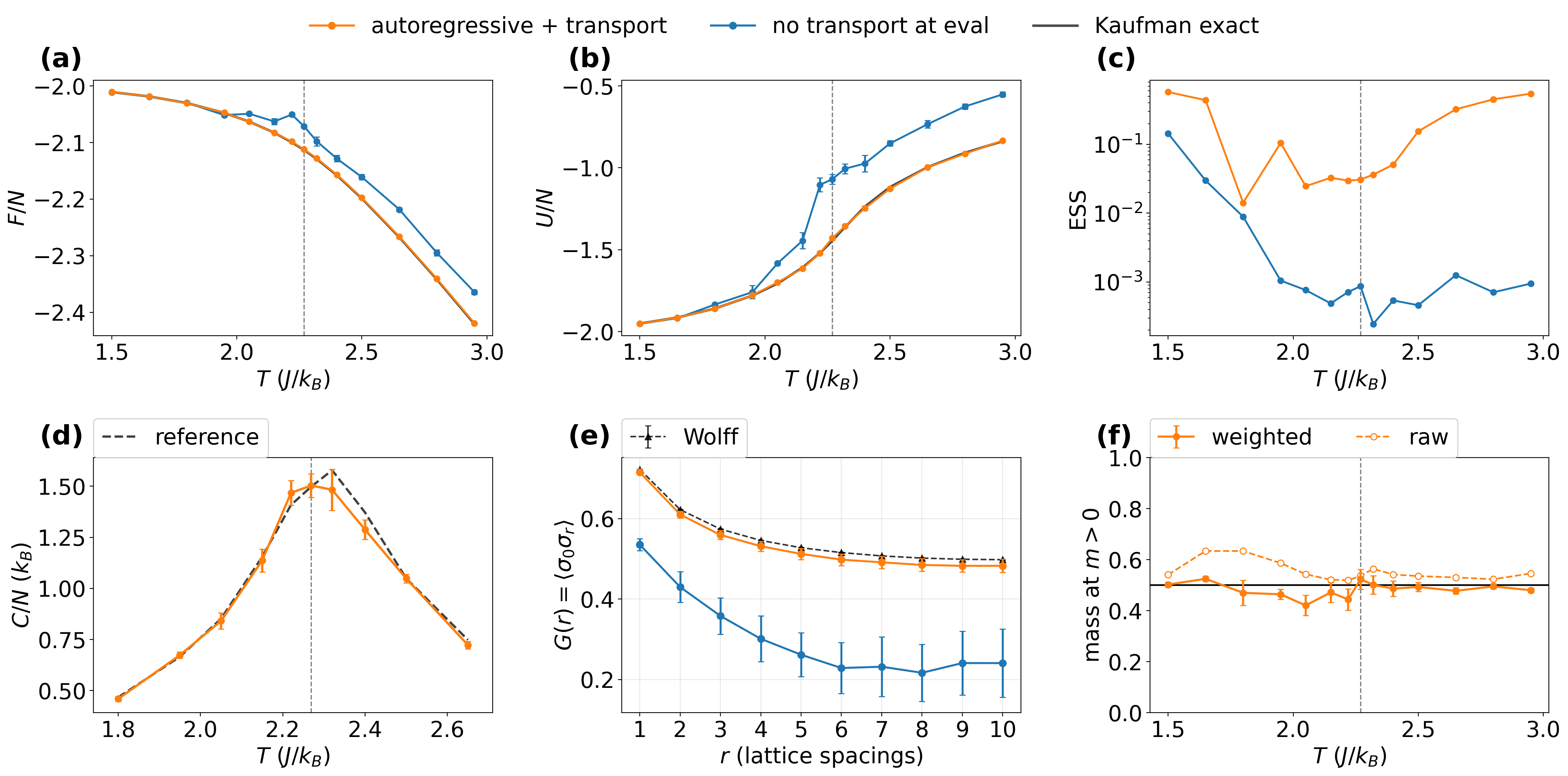}
    \caption{Ising model statistics at $\Delta\mu=0$ (i.e. $h=0$), with error bars via bootstrapping. (a) free energy per site ($F/N$) vs. $T$, compared to Kaufman exact values. (b) Energy per site ($U/N$) vs. $T$, compared to Kaufman exact values. (c) log plot of ESS vs. $T$. (d) Heat capacity per site ($C/N$) vs. $T$. The reference is the numerical derivative of the Kaufman exact energy at this temperature spacing. (e) Correlation values vs. lattice spacings, compared to the standard Wolff sampling algorithm~\citep{wolff1989collective}. (f) Weighted and unweighted (raw) sample mass at one of the two states, across different temperatures. The black line is 0.5 as the reference.}
    \label{fig:ising}
\end{figure}

In the appendix (\S\ref{sec:multiL}), we present the $20 \times 20$ system result alongside $10 \times 10$ and $15 \times 15$, following \citet{du2026scaling}, and show FrOGS matches the Kaufman/Wolff reference at every scale.

\subsection{AgPd}

AgPd, under our CE, has a single miscibility gap up to 600K (Fig.~\ref{fig:agpd}), meaning the solid solution separates into Ag-rich and Pd-rich phases up to that temperature. This is a different type of challenge from CuAu, where the phase boundary is between ordered and disordered phases. 

We reproduce the miscibility gap of a 125-site ($5\times5\times5$) AgPd model with high accuracy. We sweep the $(\Delta\mu, T)$ space, drawing 5000 samples at each point, with 0.01eV and 25K resolution in the $[-0.4\text{eV}, 0.4\text{eV}]$ and $[200\text{K}, 900\text{K}]$ range for $\Delta\mu$ and $T$, respectively, matching \citet{damewood2022sampling}.
The sampling is efficient with ESS $\ge 0.1$ in 98.94\% of the $(\Delta \mu, T)$ grid (Fig. \ref{fig:agpd_ess}) with the median value of 0.921. The shape of the binodals agreeing with the metadynamics reference verifies the accuracy of the free energies predicted by our sampler.

\subsection{CuAu}
\label{sec:cuau}

\begin{figure}
  \centering
  \begin{subfigure}[b]{0.48\linewidth}
    \includegraphics[width=\linewidth]{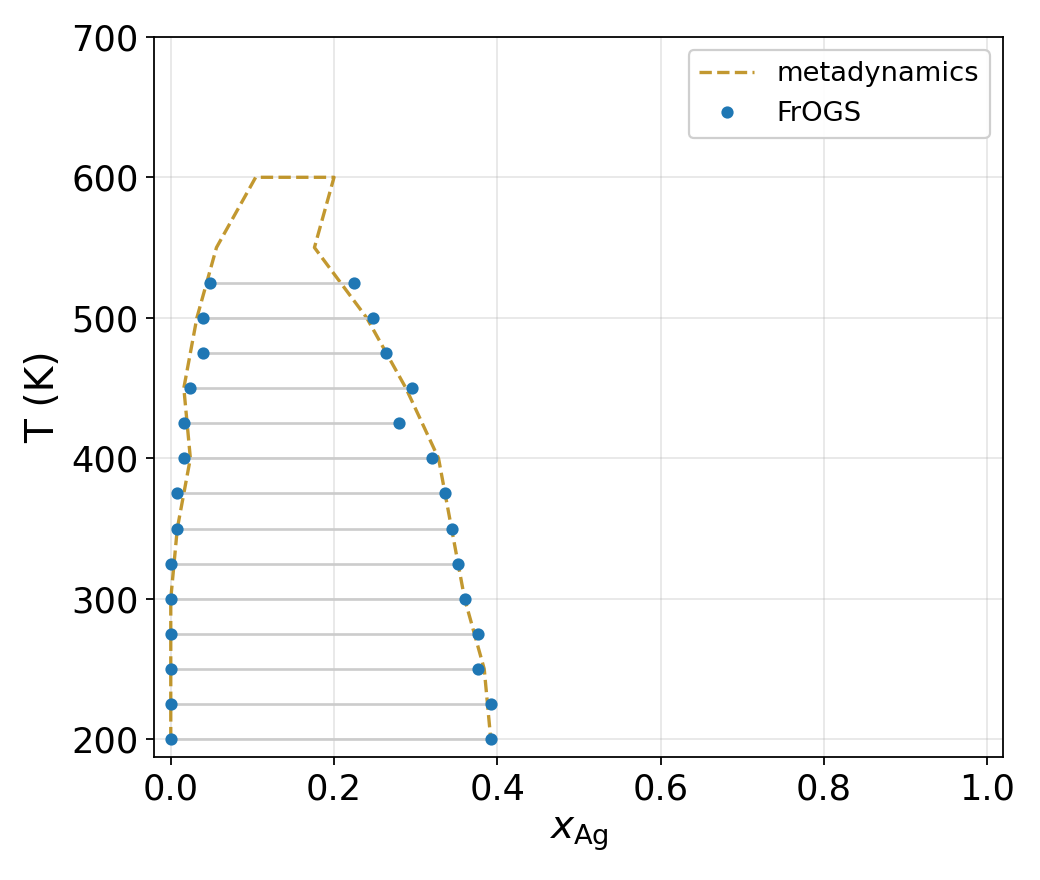}
    \caption{AgPd, FrOGS}
    \label{fig:agpd}
  \end{subfigure}
  \hfill
  \begin{subfigure}[b]{0.48\linewidth}
    \includegraphics[width=\linewidth]{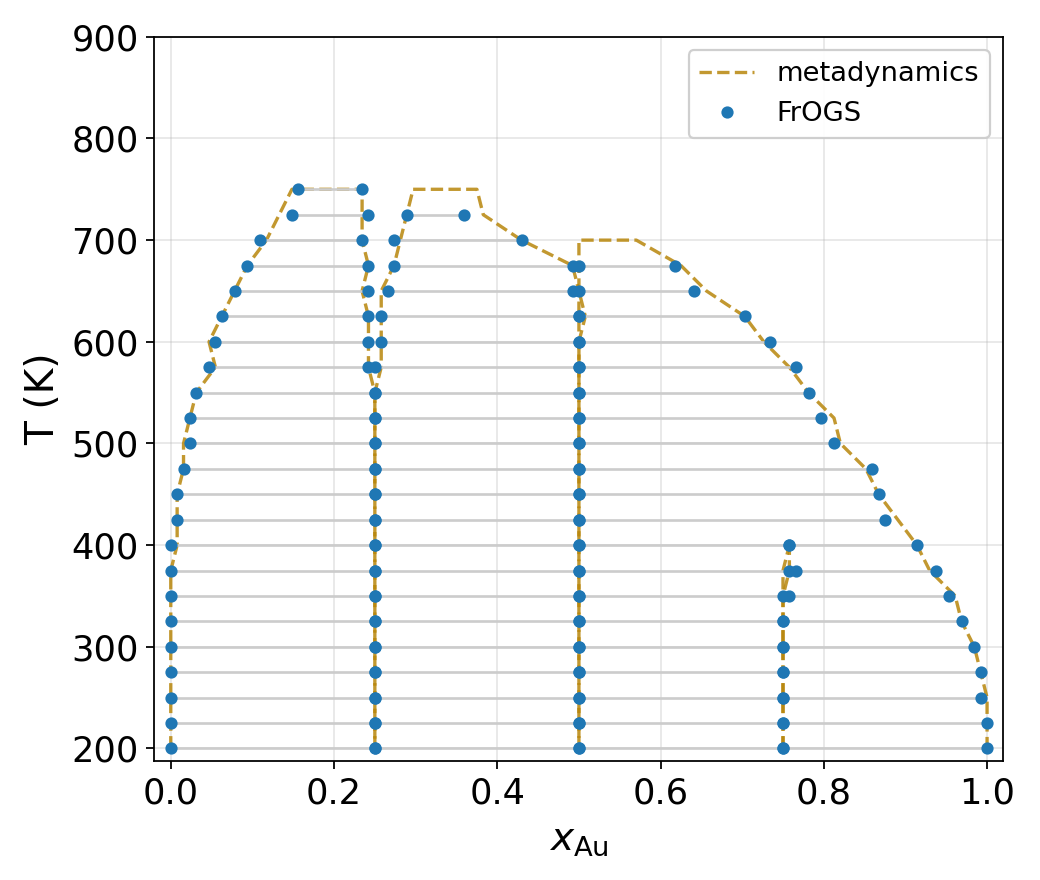}
    \caption{CuAu, FrOGS}
    \label{fig:cuau}
  \end{subfigure}

  \vspace{\baselineskip}

  \begin{subfigure}[b]{0.48\linewidth}
    \includegraphics[width=\linewidth]{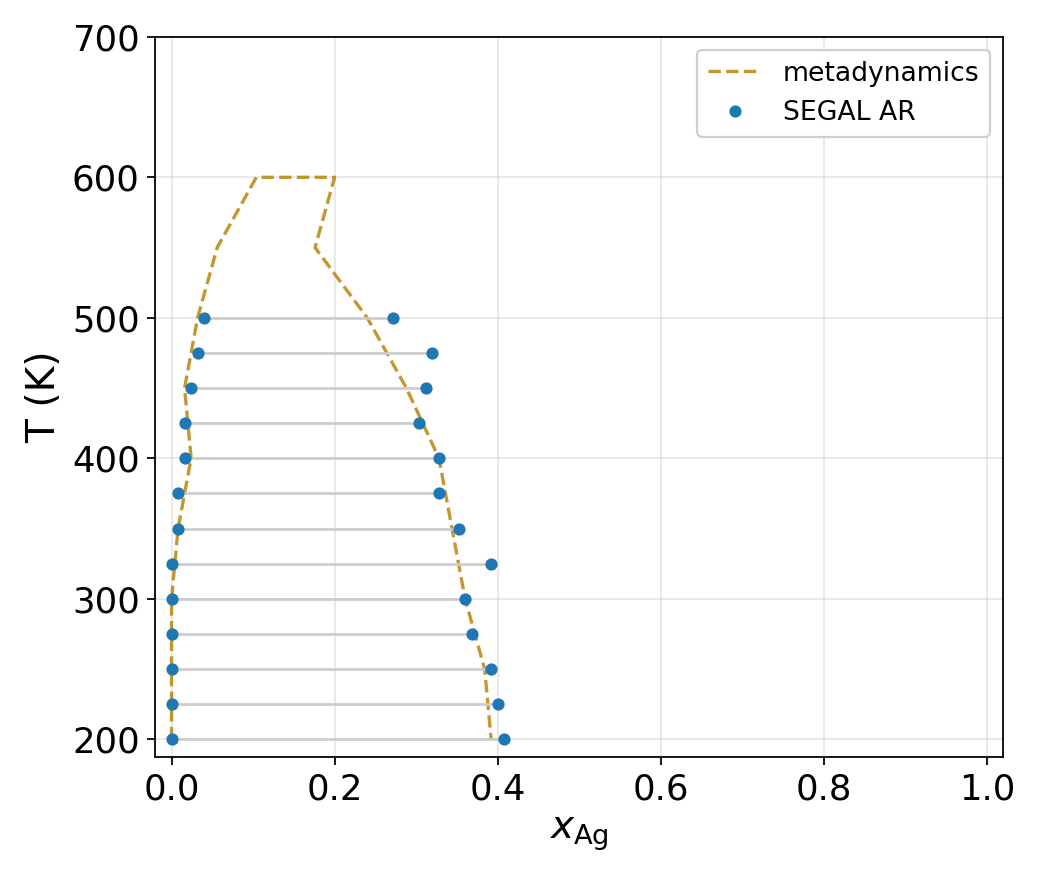}
    \caption{AgPd, SEGAL}
    \label{fig:agpd_segal}
  \end{subfigure}
  \hfill
  \begin{subfigure}[b]{0.48\linewidth}
    \includegraphics[width=\linewidth]{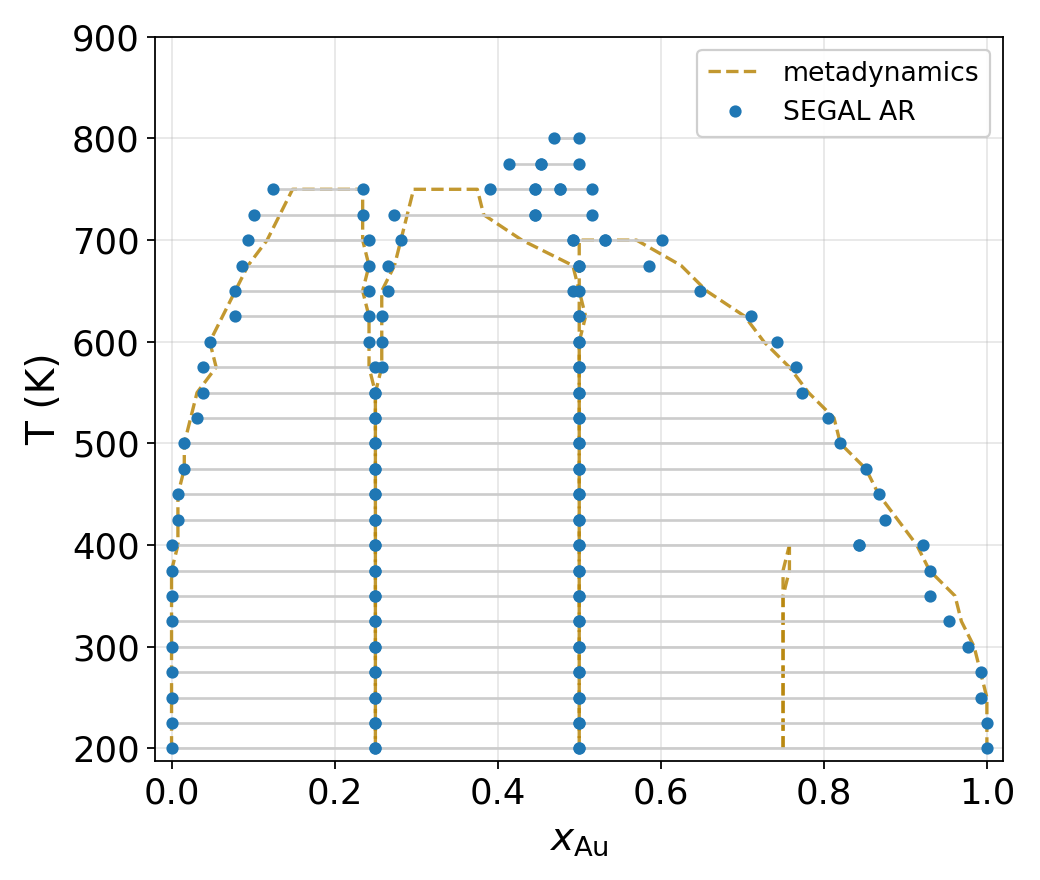}
    \caption{CuAu, SEGAL}
    \label{fig:cuau_segal}
  \end{subfigure}
  \caption{Phase diagrams of 125-site AgPd and 128-site CuAu. Top row: constructed with FrOGS. Bottom row: constructed with SEGAL, the autoregressive model of \citet{damewood2022sampling}. All four plots share the same phase diagram constructor (\S\ref{sec:phasediagram}). The metadynamics reference data is from \citet{damewood2022sampling}.}
  \label{fig:phase_diagrams}
\end{figure}

We present the phase diagram of a 128-site ($4\times4\times8$) CuAu model drawn with samples generated by the FrOGS model. We sweep the $(\Delta\mu, T)$ space, drawing 5000 samples at each point, with 0.01eV and 25K resolution in the $[-0.2\text{eV}, 0.15\text{eV}]$ and $[200\text{K}, 1200\text{K}]$ range for $\Delta\mu$ and $T$, respectively, matching \citet{damewood2022sampling} as we did for AgPd. The model reliably detects all three ordered phases $\text{Cu}_3\text{Au}$, $\text{CuAu}$, and CuAu$_3$, and the phase boundary shapes agree with the reference (Fig.~\ref{fig:cuau}).

\begin{table}[t]
  \caption{$\chi^2/\mathrm{dof}$ test against the flat $1/g$ expectation
    over the degenerate ordering variants of each phase. Samples pass the gate if their composition is within 2 sites of the
phase's ideal stoichiometry and their long-range order parameter is at
least half its value in the perfectly ordered configuration. A sweep point counts only if at least 50 samples pass. ``unweighted'' is the statistic on the bare sample counts; ``weighted'' applies the importance weights and replaces the sample size with the effective size $N_{\mathrm{eff}}=M\cdot\mathrm{ESS}\in(1,M]$.  We report the median
    over every gated point of the $(\Delta\mu,T)$ sweep (361/183/21 points, respectively) in the table.}
  \vspace{\baselineskip}
  \centering
  \begin{tabular}{lccc}
    \toprule
    $\chi^2/\mathrm{dof}$ vs $1/g$ & Cu$_3$Au ($g=4$) & CuAu ($g=6$) & CuAu$_3$ ($g=4$) \\
    \midrule
    FrOGS unweighted              & 7.8 & 23.5 & 37.6 \\
    FrOGS weighted                & 0.79 & 0.93 & 0.89 \\
    \midrule
    perfect balance (null median) & 0.79 & 0.87 & 0.79 \\
    \bottomrule
  \end{tabular}
  \label{tab:chi2}
\end{table}

FrOGS distributes samples over all of the degenerate modes of each ordered phase without mode collapse. We quantify this through a $\chi^2/\mathrm{dof}$ test against a uniform distribution. The raw samples deviate from uniform, while the importance-weighted counts are
close to uniform (Table~\ref{tab:chi2}), demonstrating samples are spread correctly across configurations. The importance weights are therefore correcting an imbalance
among configurations of the same energy, using information about the full configuration. A near-perfectly trained model will yield
similar values to the last row for the raw counts as well, so the gap between the top two rows shows how much the importance weights are correcting. 

The ESS figure (Fig.~\ref{fig:cuau_ess}) elucidates the accuracy of the phase diagram drawn with FrOGS. High ESS means the variance of the log-weights ($A$) is low, so the samples contribute more evenly to the estimation of the partition function~\eqref{eq:partition_t}. With $\text{ESS} \geq 0.1$ in 99.93\% of the $(\Delta \mu, T)$ grid (1475 of 1476 conditions) with the median value of 0.917, the phase diagram constructor can reliably estimate the free energy.

\begin{figure}
  \centering
  \begin{subfigure}[b]{0.48\linewidth}
    \includegraphics[width=\linewidth]{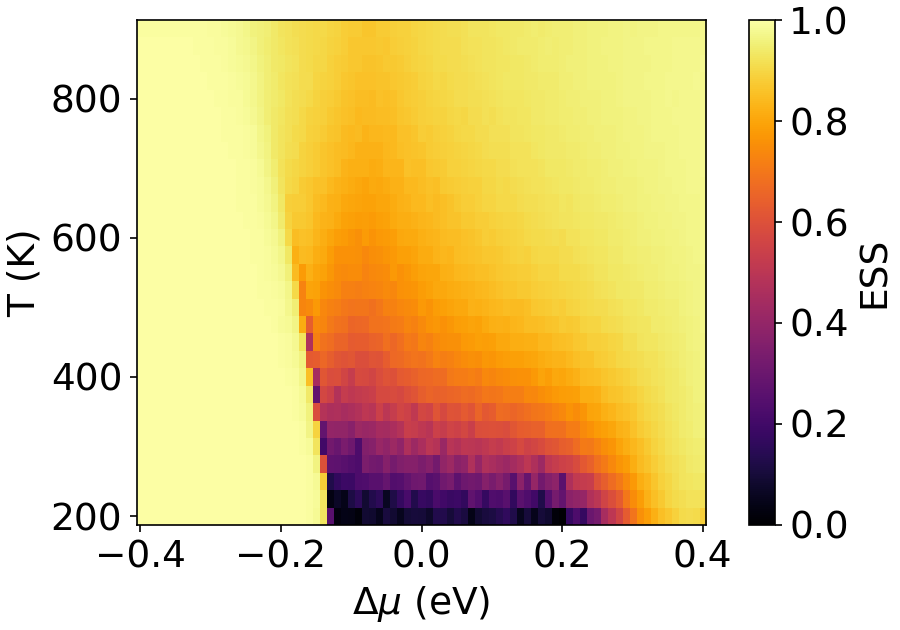}
    \caption{AgPd}
    \label{fig:agpd_ess}
  \end{subfigure}
  \hfill
  \begin{subfigure}[b]{0.48\linewidth}
    \includegraphics[width=\linewidth]{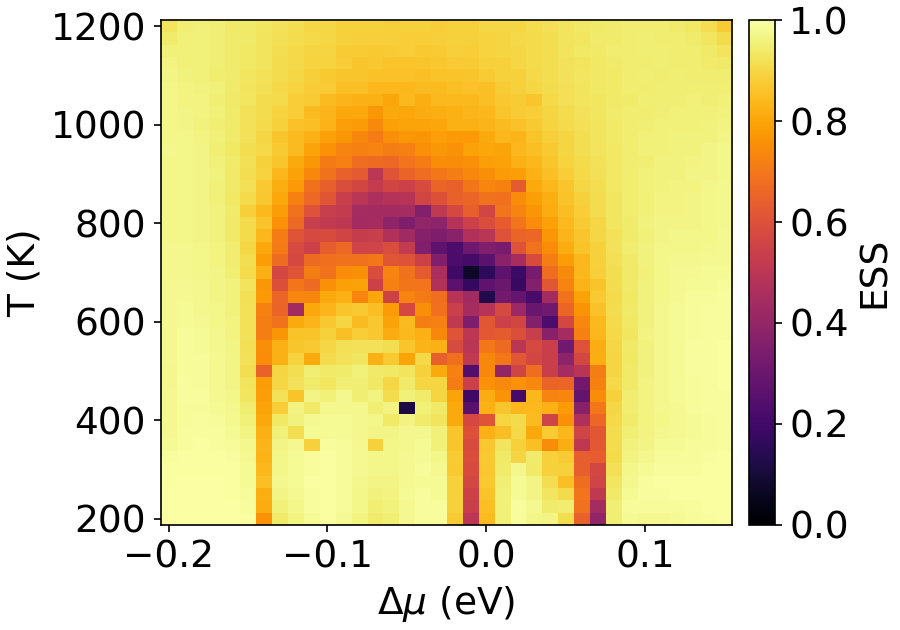}
    \caption{CuAu}
    \label{fig:cuau_ess}
  \end{subfigure}
  \caption{ESS plots generated with FrOGS for 125-site AgPd and 128-site CuAu.}
  \label{fig:ess}
\end{figure}

\subsection{Comparison with previous work}
\label{sec:comp}

We benchmark our work against two previous samplers, \citet{damewood2022sampling} and \citet{du2026scaling}. We rerun our phase diagram constructor (\S\ref{sec:phasediagram}) with the released checkpoint by \citet{damewood2022sampling}. \citet{du2026scaling} does not have public code as of November 2026, so we compare our phase diagram qualitatively.

First, we observe that only FrOGS reliably captures the stability range of the CuAu$_3$ phase. We observe this by comparing Fig.~\ref{fig:cuau} against Fig.~\ref{fig:cuau_segal} and \citet[Fig.~4g]{du2026scaling}. We also do a sensitivity analysis in \S\ref{sec:sensitivity} to show that the comparison of Fig.~\ref{fig:cuau} to Fig.~\ref{fig:cuau_segal} holds across different phase diagram construction thresholds and different initialization seeds for FrOGS. Since the formation energy of $\text{CuAu}_3$ is only slightly favored compared to the $\text{CuAu}+\text{Au}$ mix~\citep[Fig.~3a]{damewood2022sampling}, the phase is thermodynamically unstable above 400K, according to metadynamics. The phase diagram drawn with the FrOGS samples agreeing with the metadynamics result for $\text{CuAu}_3$ demonstrates the accuracy of our framework. 

Furthermore, FrOGS tracks the shape of the phase boundary curves more closely to the metadynamics reference than \citet{damewood2022sampling} (Fig.~\ref{fig:phase_diagrams}), and FrOGS better agrees with the reference on the heat capacity of the Ising model than \citet[Fig.~3d]{du2026scaling}. Higher ESS values across different chemical conditions (Fig.~\ref{fig:cuau_ess}) lead FrOGS to predict the phase boundary curves more accurately.

\section{Conclusion}
\label{sec:conclusion}

To our knowledge, FrOGS is the first discrete neural sampler conditioned on $(\Delta\mu, T)$ that combines an autoregressive network with a learned CTMC transport. It provides i.i.d.\ configurations and unbiased estimates of the partition function, placing free energies at every condition on a common absolute scale without an external reference. To confirm the validity of FrOGS, we first compare its prediction of the thermodynamic variables of the 2D Ising model to established references. We then benchmark FrOGS on generating phase diagrams of two well-established chemical systems, AgPd and CuAu, and show that it more faithfully reconstructs these phase diagrams than previous autoregressive approaches. The next natural step would be to apply our framework to challenges in materials science such as the simulation of multicomponent ($K\geq3$) alloys.

Conventional MCMC requires thermodynamic integration that gives each phase its own chain, started from an ordering identified beforehand, since the phase-boundary crossing time grows exponentially in the interfacial area $L^{d-1}$~\citep{berg1991multicanonical}. The microstates must be reached without crossing a boundary~\citep{vandeWalle2002mc, frenkel1984new}, which is hard to guarantee as $K$ grows. FrOGS instead generates i.i.d.\ samples from the CE alone, and the free energy is directly computable at any chemical condition.

FrOGS starts from a random prior with a randomly initialized CNN and is optimized from there. A cold start suffices at this scale, which is encouraging; for larger cells, a physically motivated prior (e.g. a mean-field construction) is a natural next step. On one NVIDIA A100, FrOGS takes $5.0$ hours to train and $22$ hours to sweep the CuAu grid at $53$ s per condition for $M=5000$ samples. This compares to a few core-hours for an anchored Metropolis calculation on the same system (\S\ref{sec:compute}). However, Metropolis sampling uses the ordered structure of each phase to initialize the calculation, whereas FrOGS does not require this knowledge. Moreover, a single trained FrOGS model covers multiple chemical conditions, so we can sample only at the conditions of interest. In multicomponent alloy design, these conditions form a sparse set of points in a high-dimensional composition simplex. Classical methods in this regime scale poorly. For example, Wang--Landau sampling requires a $K$-dimensional density of states
converged everywhere, and metadynamics must fill a $(K-1)$-dimensional collective variable space, whose cost grows exponentially with $K$~\citep{Bussi2020Metadynamics}. Finally, because we can compute the gradients of the free energy with the same samples via standard thermodynamic identities, FrOGS opens the way to
gradient-based inverse design.

\section*{Code and Data Availability}
The code and data used in this work are available at \url{https://github.com/kcmin02/FrOGS}.

\begin{ack}
We thank Peter Holderrieth, Jason Ogbebor, Daniel Xiao, Ameya Daigavane, and Ojas Gulati for helpful discussions regarding our methodology.

K.M. and T.E.S. were supported by the Air Force Office of Scientific Research under Award No. FA9550-24-1-0067. K.M. would like to acknowledge support from the Kwanjeong Educational Foundation. E.H. was supported by the U.S. Department of Energy, Office of Science, Office of Advanced Scientific Computing Research, Department of Energy Computational Science Graduate Fellowship under Award Number DE-SC0024386. This research used resources of the MIT Office of Research Computing and Data.

This report was prepared as an account of work sponsored by an agency of the United States Government. Neither the United States Government nor any agency thereof, nor any of their employees, makes any warranty, express or implied, or assumes any legal liability or responsibility for the accuracy, completeness, or usefulness of any information, apparatus, product, or process disclosed, or represents that its use would not infringe privately owned rights. Reference herein to any specific commercial product, process, or service by trade name, trademark, manufacturer, or otherwise does not necessarily constitute or imply its endorsement, recommendation, or favoring by the United States Government or any agency thereof. The views and opinions of authors expressed herein do not necessarily state or reflect those of the United States Government or any agency thereof.
\end{ack}

\newpage
\bibliographystyle{abbrvnat}
\bibliography{references}

\newpage
\appendix

\section{Technical appendices and supplementary material}

\begingroup
\parindent=0pt \parskip=2pt
\S\ref{sec:sgc}\quad Sampling from the SGC ensemble\dotfill\pageref{sec:sgc}\par
\S\ref{sec:leapsth}\quad LEAPS\dotfill\pageref{sec:leapsth}\par
\S\ref{sec:derivation}\quad Deriving the reparametrized corrector\dotfill\pageref{sec:derivation}\par
\S\ref{sec:phasediagram}\quad Phase diagram construction\dotfill\pageref{sec:phasediagram}\par
\S\ref{sec:gradient}\quad Joint optimization of the autoregressive network\dotfill\pageref{sec:gradient}\par
\S\ref{sec:vargrad}\quad The prior objective as the VarGrad generalized\dotfill\pageref{sec:vargrad}\par
\S\ref{sec:proof}\quad Proof of the local equivariance of the U-net pyramid\dotfill\pageref{sec:proof}\par
\S\ref{sec:film}\quad The FiLM setup\dotfill\pageref{sec:film}\par
\S\ref{app:reindex}\quad Re-indexing the fcc lattice onto an integer grid\dotfill\pageref{app:reindex}\par
\S\ref{sec:ablation}\quad Ablation study of prior-only and transport-only networks\dotfill\pageref{sec:ablation}\par
\S\ref{sec:training}\quad Model and training details\dotfill\pageref{sec:training}\par
\S\ref{sec:multiL}\quad Multiscale 2D Ising result comparison\dotfill\pageref{sec:multiL}\par
\S\ref{sec:sensitivity}\quad Sensitivity analysis of FrOGS phase count\dotfill\pageref{sec:sensitivity}\par
\S\ref{sec:compute}\quad FrOGS vs.\ conventional Monte Carlo\dotfill\pageref{sec:compute}\par
\endgroup

\subsection{Sampling from the SGC ensemble} 
\label{sec:sgc}
Our goal is to sample from the semi-grand canonical ensemble. Consider a fixed lattice of $N$ sites indexed by $i \in \{1,\dots,N\}$. For a $K$-component alloy, each site carries an occupation variable $\sigma_i \in \{1,\dots,K\}$ labeling the chemical species on it. A configuration is $\sigma = (\sigma_1,\dots,\sigma_N)$ drawn from the configuration space $\Omega = \{1,\dots,K\}^N$ with $|\Omega| = K^N$. The species counts and concentrations are
\begin{equation}
  N_\alpha(\sigma) = \sum_{i=1}^{N} \delta_{\sigma_i,\alpha},
  \qquad
  x_\alpha(\sigma) = \frac{N_\alpha(\sigma)}{N},
  \qquad
  \sum_{\alpha=1}^{K} N_\alpha(\sigma) = N ,
\end{equation}
where $\delta$ is the Kronecker delta. The configurational energy is supplied by
a cluster expansion fit to first-principles data,
\begin{equation}
  E(\sigma) = \sum_{\omega} J_\omega\,\Phi_\omega(\sigma),
  \label{eq:ce}
\end{equation}
with effective cluster interactions $J_\omega$ and cluster correlation functions
$\Phi_\omega$. For our purposes, $E : \Omega \to \mathbb{R}$ is simply a
scalar function of the configuration that is cheap to evaluate.

In the SGC ensemble, the number of sites $N$ and the temperature $T$ are held fixed while the composition fluctuates through exchange with a reservoir set by chemical potentials $\mu_\alpha$. Because $\sum_\alpha N_\alpha = N$ is constrained, only the chemical potential differences are thermodynamically meaningful. Taking species $K$ as the reference, we define
$\Delta\mu_\alpha = \mu_\alpha - \mu_K$ for $\alpha = 1,\dots,K-1$ and collect
them as $\Delta\mu = (\Delta\mu_1,\dots,\Delta\mu_{K-1})$. A thermodynamic
condition is then the pair $c = (\Delta\mu, T)$, and with
$\beta = 1/(k_{\mathrm{B}} T)$ the semi-grand canonical partition function is
\begin{equation}
  Z_{\mathrm{SG}}(c)
  = \sum_{\sigma\in\Omega}
    \exp\!\Big[
      -\beta\Big(
        E(\sigma)
        - \sum_{\alpha=1}^{K-1} \Delta\mu_\alpha\, N_\alpha(\sigma)
      \Big)
    \Big],
  \label{eq:zsg}
\end{equation}
where an additive constant $\beta\mu_K N$ common to every configuration has been
absorbed into the free-energy reference. Introducing the effective Hamiltonian
\begin{equation}
  \mathcal{H}_c(\sigma)
  = E(\sigma)
    - \sum_{\alpha=1}^{K-1} \Delta\mu_\alpha\, N_\alpha(\sigma),
    \label{eq:hamiltonian}
\end{equation}
the equilibrium distribution we wish to sample is the Boltzmann distribution
\begin{equation}
  p_c(\sigma)
  = \frac{1}{Z_{\mathrm{SG}}(c)}\,
    e^{-\beta \mathcal{H}_c(\sigma)},
  \qquad
  Z_{\mathrm{SG}}(c) = \sum_{\sigma\in\Omega}
    e^{-\beta \mathcal{H}_c(\sigma)}.
  \label{eq:target}
\end{equation}
For a binary alloy ($K=2$; we use $\sigma_i \in \{-1,1\}$ instead of $\{1, 2\}$), this reduces to $\mathcal{H}_c(\sigma) = E(\sigma) -
\Delta\mu \sum_i \frac{\sigma_i+1}{2}$, with $\Delta\mu$ as a scalar.

The thermodynamic potential, which we refer to as the free energy, is
\begin{equation}
  F_{\mathrm{SG}}(c) = -k_{\mathrm{B}} T \, \ln Z_{\mathrm{SG}}(c).
  \label{eq:freeenergy}
\end{equation}
A general observable follows from
$\langle O \rangle_c = \sum_{\sigma} O(\sigma)\, p_c(\sigma)$. 

$Z_{\mathrm{SG}}(c)$ is a sum over $K^N$ configurations and is intractable to evaluate directly. A sampler that both draws from $p_c$ and yields an unbiased estimate of $Z_{\mathrm{SG}}(c)$ across a range of conditions from a single trained model would be beneficial, which is what we aim to provide.

\subsection{LEAPS}
\label{sec:leapsth}

The following is a brief summary of the theory of LEAPS \citep{holderrieth2025leaps} concerning our research.

The rate matrix $Q^\theta_t$ (\S\ref{sec:LEAPS}) transports $\rho_0$ exactly along the path iff it solves
the Kolmogorov forward equation (KFE):
\begin{equation}
\partial_t\rho_t(\sigma)=\sum_{i,\,\tau\neq\sigma_i}\!\Big(
        Q^\theta_t(\sigma_i,i\mid\sigma')\,
        \rho_t(\sigma')
        -Q^\theta_t(\tau,i\mid\sigma)\rho_t(\sigma)
      \Big),
    \quad \sigma'=\mathrm{Swap}(\sigma,i,\tau),
\end{equation}
which an arbitrary network will not. LEAPS records the mismatch
within the operator $\mathcal{K}^\theta_t$:
\begin{equation}
  \mathcal{K}^\theta_t \rho_t(\sigma)
    = -\,\partial_t U_t(\sigma)
    - \sum_{i,\,\tau\neq\sigma_i}\!\Big(
        Q^\theta_t(\sigma_i,i\mid\sigma')\,
        \tfrac{\rho_t(\sigma')}{\rho_t(\sigma)}
        -Q^\theta_t(\tau,i\mid\sigma)
      \Big),
    \quad \sigma'=\mathrm{Swap}(\sigma,i,\tau).
  \label{eq:corrector}
\end{equation}
$\mathcal{K}^\theta_t$ is the KFE mismatch minus the change in free energy, since
$(\text{LHS}-\text{RHS of KFE})/\rho_t
= \partial_t \log \rho_t - Q_t^\theta \rho_t/\rho_t
= \partial_t F_t - \partial_t U_t - Q_t^\theta \rho_t/\rho_t
= \partial_t F_t + \mathcal{K}^\theta_t \rho_t$, with $Q_t^\theta \rho_t$ shorthand
for the right-hand side of the KFE.
 
Each sample integrates this residual along its own trajectory into a log-weight
\begin{equation}
  A_t=\int_0^t \mathcal{K}^\theta_s\rho_s(X_s)\,ds,
  \qquad A_0=0,
  \label{eq:weight}
\end{equation}
analogous to the accumulated work in the Jarzynski equality, given
$\mathcal{K}^\theta_t \rho_t(\sigma) = -\partial_tF_t$ at perfect transport.
Crucially, the discrete-state Jarzynski equality holds for any
$Q^\theta_t$~\citep[Thm.~5.2]{holderrieth2025leaps},
\begin{equation}
    \mathbb{E}[e^{A_t}]=Z_t/Z_0 ,
    \label{eq:partition}
\end{equation}
and any observable is recovered by reweighting,
$\langle O\rangle_{\rho_t}=\mathbb{E}[e^{A_t}O(X_t)]/\mathbb{E}[e^{A_t}]$. This
gives us the two properties FrOGS is built on: a direct estimate of
$Z_{\mathrm{SG}}=Z_1$ and consistent observables.
 
Since $\mathrm{Var}[A_t]=0$ exactly when $X_t\sim\rho_t$, the transport is trained
by minimizing the weight variance. LEAPS minimizes the residual of the KFE with a
physics-informed neural network (PINN),
\begin{equation}
    \mathcal{L}(\theta,\phi;t)
    = \mathbb{E}_{s\sim \mathrm{Unif}_{[0,t]},\, \sigma_s\sim\mathcal{R}_s}
       \Big[\big|\mathcal{K}_s^{\theta}\rho_s(\sigma_s) + \partial_s F_s^{\phi}\big|^2\Big],
    \qquad
    \mathcal{L}^{\text{log-var}}(\theta;t)\le t^2\,\mathcal{L}(\theta,\phi;t),
    \label{eq:pinn}
\end{equation}
where $\mathcal{L}^{\text{log-var}}(\theta;t)= \text{Var}[A_t]$. $(\theta^*,\phi^*)$, the unique minimizer of $\mathcal{L}$, has $Q_t^{\theta^*}$ satisfying the KFE
and $F_t^{\phi^*}=F_t$~\citep[Prop.~6.1]{holderrieth2025leaps}. Here
$F^\phi_t:\mathbb{R}\to\mathbb{R}$ ($\mathbb{R}^3\to\mathbb{R}$ in our case, with $\Delta\mu$ and $T$ additionally as inputs) is a free-energy network, a scalar function of
time that tracks the $\partial_t F_t$ part of $-\mathcal{K}_t^\theta$; it is
a training-time device and enters neither $A_t$ nor the free-energy readout. The
reference measure $\mathcal{R}$ is arbitrary provided its support covers the path,
so states may be drawn from a replay buffer rather than fresh simulations.

The remaining obstacle is computational. Evaluating Eq.~\eqref{eq:corrector} requires
the reverse rates at all $N(K-1)$ single-site neighbors of $\sigma$. LEAPS solves this issue with a locally equivariant parametrization of $Q^\theta_t$ (\S\ref{sec:arch})
that yields the whole sum in a single forward pass.

\subsection{Deriving the reparametrized corrector}
\label{sec:derivation}

Without loss of generality, we assume a fixed condition $c$. For a move
$\sigma'=\mathrm{Swap}(\sigma,i,\tau)$, let
$\Delta=U_t(\sigma')-U_t(\sigma)$,
$g=G_t^\theta(\tau,i\mid\sigma)$, and
$M=M_t(\tau,i\mid\sigma)=e^{-[\Delta]_+}$.
The reverse move has energy difference $-\Delta$ and Metropolis factor
$M'=e^{-[-\Delta]_+}$. Hence
\begin{equation}
  M'e^{-\Delta}=M,
  \label{eq:db}
\end{equation}
by $[-\Delta]_++\Delta=[\Delta]_+$.

\begin{proposition}
\label{prop:reparam}
For the rates $\tilde Q_t^\theta=[G_t^\theta]_+M_t$ of ~\eqref{eq:reparam}, the corrector~\eqref{eq:corrector} becomes
\begin{equation}
  \tilde{\mathcal{K}}_t^\theta\rho_t(\sigma)
  =-\partial_tU_t(\sigma)
   +\sum_{i,\,\tau\neq\sigma_i}
      M_t(\tau,i\mid\sigma)G_t^\theta(\tau,i\mid\sigma),
\end{equation}
with $|M_tG_t^\theta|\le |G_t^\theta|$.
\end{proposition}

\begin{proof}
Local equivariance~\eqref{eq:le} gives the reverse rate
$\tilde Q_t^\theta(\sigma_i,i\mid\sigma')=[-g]_+M'$.
Since $\rho_t(\sigma')/\rho_t(\sigma)=e^{-\Delta}$, each term in the
sum of Eq.~\eqref{eq:corrector} reduces to
\begin{equation}
  [-g]_+M'e^{-\Delta}-[g]_+M
  =M\big([-g]_+-[g]_+\big)=-Mg.
\end{equation}
The bound follows from $0<M_t\le1$.
\end{proof}

The corrector therefore requires one evaluation of the flux network and no
standalone density ratio $e^{-\Delta}$. Its Metropolis factor has a
non-positive exponent, avoiding the overflow that can occur for strongly
downhill moves. For the transport parameters $\theta$, $M_t$ is fixed, so
$|\partial_\theta(M_tG_t^\theta)|
\le |\partial_\theta G_t^\theta|$ as well.

\paragraph{Unbiasedness.}
The reparametrized off-diagonal rates are nonnegative. Using these rates with
the matching corrector in Proposition~\ref{prop:reparam} preserves the
continuous-time LEAPS identity
$\mathbb{E}[e^{A_t}]=Z_t/Z_0$~\citep[Thm.~5.2]{holderrieth2025leaps}.

\paragraph{Expressiveness.}
Call a rate matrix \emph{one-way} if at most one direction of each
configuration pair has a nonzero rate. Both $[G]_+$ and $[G]_+M_t$ have this
property when $G$ is antisymmetric.

\begin{proposition}
\label{lem:express}
For every one-way rate matrix $Q^\star$, there is an antisymmetric flux $G$
such that $[G]_+M_t=Q^\star$. Thus, at the level of unrestricted
antisymmetric fluxes, the reparametrization preserves the one-way rate
family underlying the universal-representation result of
\citet[Prop.~8.1]{holderrieth2025leaps}.
\end{proposition}

\begin{proof}
Because $M_t>0$, we can define
\begin{equation}
  G(\tau,i\mid\sigma)
  =\frac{Q^\star(\tau,i\mid\sigma)}{M_t(\tau,i\mid\sigma)}
   -\frac{Q^\star(\sigma_i,i\mid\sigma')}{M_t(\sigma_i,i\mid\sigma')}.
\end{equation}
Reversing the move exchanges the terms and negates $G$. If the forward rate
is positive, the reverse rate is zero, so $[G]_+M_t=Q^\star$.
If the forward rate is zero, $G\le0$ and the same identity holds.
Conversely, $[G]_+M_t$ is one-way for every antisymmetric $G$.
\end{proof}

The Metropolis factor also changes the parametrization's inductive bias by
explicitly damping uphill moves. We observe a faster early decrease of the
PINN loss with the reparametrized rates.

\subsection{Phase diagram construction}
\label{sec:phasediagram}

We describe the phase diagram constructor which yields the entire phase diagram with the samples from a single trained FrOGS model. For this process, we only need for each sample
$m$ its log-weight $A^{(m)}$~\eqref{eq:weight} and its composition. No information about the phases is supplied manually.

The learned prior is exactly normalized ($Z_0=1$, \S\ref{sec:phasemix}), so combining
Eq.~\eqref{eq:partition_t} with Eq.~\eqref{eq:reweighting} gives
$\mathbb{E}\big[e^{A}\,g(X_1)\big]
=Z_{1}(c)\,\langle g\rangle_{p_c}$ for any function $g$ of the final
configuration. Take $g$ to be the indicator that the configuration contains
exactly $n$ atoms of species $1$. Every configuration in that class carries the
same factor $e^{\beta\Delta\mu\,n}$, so it comes outside the sum:
\begin{equation}
  \mathbb{E}_{\rho_0}\big[e^{A}\,\mathbf{1}_{\{N_1=n\}}\big]
    \;=\;Z_{1}(c)\,p_c(N_1\!=\!n)
    \;=\;e^{\beta\Delta\mu\,n}\,Z_{\mathrm{c}}(n,T),
  \qquad
  Z_{\mathrm{c}}(n,T)\equiv\!\!\!\sum_{\sigma:\,N_1(\sigma)=n}\!\!\!e^{-\beta E(\sigma)} .
  \label{eq:pd-canonical}
\end{equation}
The left-hand side is estimated without bias by the sample average
$\frac{1}{M}\sum_m e^{A^{(m)}}\mathbf{1}_{\{N_1^{(m)}=n\}}$. The right-hand side
is the partition function at a fixed composition times a factor. Dividing by $e^{\beta\Delta\mu\,n}$ therefore gives an
unbiased estimator of $Z_c$ on the integer lattice $n=0,\dots,N$, with $x=n/N$:
\begin{equation}
  \widehat{Z}_{\mathrm{c}}(n,T)
    =\frac{e^{-\beta\Delta\mu\,n}}{M}
     \sum_{m\,:\,N_1^{(m)}=n} e^{A^{(m)}},
  \qquad
  f(x,T)=-\frac{k_{\mathrm{B}}T}{N}\,\ln Z_{\mathrm{c}}(n,T).
  \label{eq:pd-fx}
\end{equation}
This is a histogram, composition being discrete by
construction. In practice, the sum is accumulated with the largest $A^{(m)}$ factored out.

Because $Z_{\mathrm{c}}(n,T)$ carries no $\Delta\mu$-dependence, every $\Delta\mu$
at a given temperature yields an estimate of the same curve $f(x,T)$, each
from an independent set of samples. Let $\hat y_i$ be the estimate of
$\ln Z_{\mathrm{c}}(n,T)$ from the $i$-th condition. Since $\widehat{Z}_{\mathrm{c}}$ \eqref{eq:pd-fx} is a sum of i.i.d.
weights, the delta method gives the
standard error of $\hat y_i$ as
\begin{equation}
  s_i = N_{\mathrm{eff},i}(n)^{-1/2},
  \qquad
  N_{\mathrm{eff},i}(n)
    = \frac{\Big(\sum_{m\,:\,N_1^{(m)}=n} e^{A^{(m)}}\Big)^{2}}
           {\sum_{m\,:\,N_1^{(m)}=n} e^{2A^{(m)}}}.
  \label{eq:pd-se}
\end{equation}
Sums are over the samples drawn at condition $i$ that land in bin $n$.
This is the effective sample count restricted to the bin.

Since condition $i$ places most of its samples near its own equilibrium
composition $x(\Delta\mu_i,T)$, $s_i$ is smallest for bins near that composition
and grows with distance from it; bins the condition never reaches contribute
nothing. Estimates that are unbiased\footnote{To be precise, $\hat y_i$ is the log of an unbiased estimator, so it carries a
Jensen bias $-1/(2N_{\mathrm{eff},i})+O(N_{\mathrm{eff},i}^{-2})$ propagating to $-k_n/(2\sum_i w_i)$ in $\hat y$. The leading order term of the bias stays below the reported
standard error if the mean $N_{\mathrm{eff},i}$ over contributing
conditions exceeds $k_n/4$.} and independent differ only in variance, so
the minimum-variance way to merge them is to weight each by its precision:
\begin{equation}
  \hat y=\frac{\sum_i w_i\,\hat y_i}{\sum_i w_i},
  \qquad
  \mathrm{Var}\big[\hat y\big]=\frac{1}{\sum_i w_i},
  \qquad
  w_i=s_i^{-2}.
  \label{eq:pd-ivw}
\end{equation}
Let $k_n$ be the number of conditions with $w_i>0$. Typically $k_n>1$, and
\begin{equation}
  \chi^2_\nu(n)=\frac{1}{k_n-1}\sum_i w_i\,\big(\hat y_i-\hat y\big)^2
  \label{eq:pd-chi2}
\end{equation}
has expectation $\approx 1$ when the $s_i$ are accurate. The $k_n$ conditions
independently estimate the same $\hat y$, and their estimates should then
scatter about it by just the amount the weights predict. This statistic thus
cross-checks the accuracy of the $s_i$. We scale every error bar at that
temperature by
$\lambda=\max\{1,\ \mathrm{median}_{n\,:\,\sum_i w_i(n)\ge W}\ \chi^2_\nu(n)\}^{1/2}$
with $W=5$~\citep{Birge1932, ParticleDataGroup2024}. The median keeps $\lambda$ insensitive to a few discrepant bins,
and the one-sided maximum means the correction can only widen an error bar.

Two phases coexist when one line is tangent to $f(x)$ at two
compositions $x_a,x_b$. Such a line shares with both phases its slope, which is
the chemical potential $\Delta\mu$, and its intercept at $x=0$, which is the
semi-grand potential per site $F_{\mathrm{SG}}/N$
(Eq.~\eqref{eq:freeenergy}). Any overall
composition between $x_a$ and $x_b$ is then realized as a non-homogeneous mixture of the two phases, which is why the interval is called a
two-phase region and $x_a,x_b$ the phase boundaries. We accept a hull edge as a tie-line only if (i)~it skips at least four
compositions, (ii)~the measured $f(x)$ lies at least three standard errors above
it somewhere in between, or no sample at any condition ever visits those
compositions (evidence of a gap), and (iii)~the same gap appears at the temperature below.

For AgPd, the hull has a single tie-line at every temperature below the critical temperature, with one miscibility
gap, giving the Ag-rich and Pd-rich phase boundaries together with
$\Delta\mu_{\mathrm{coex}}(T)$ and closing at a critical point
(Fig.~\ref{fig:agpd}). For CuAu several tie-lines coexist, and the construction returns
the ordered phases $\mathrm{Cu_3Au}$, CuAu and $\mathrm{CuAu_3}$ at their tangent
compositions (Fig.~\ref{fig:cuau}). A tie-line closes where the two phases' compositions merge, which for phases near the same stoichiometry can happen before their
atomic ordering is actually lost. The resulting temperatures are therefore not by themselves order--disorder transition temperatures, in the most precise sense.

\subsection{Joint optimization of the autoregressive network}
\label{sec:gradient}
The annealing path is the linear
  interpolation
  \begin{equation}
    U_t(\sigma\mid c) = (1-t)\,U_0(\sigma\mid c) + t\,U_1(\sigma\mid c),
    \qquad
    U_0 = -\log\rho^{\psi}_0(\sigma\mid c),
    \quad
    U_1 = \beta\,\mathcal{H}_c(\sigma),
    \label{eq:linpath}
  \end{equation}
  where $\psi$ stands for the weights of the autoregressive network. Both ingredients of the
  corrector~\eqref{eq:corrector} then carry an explicit $\psi$-dependence:
  \begin{equation}
    \partial_t U_t(\sigma\mid c)
    = \beta\,\mathcal{H}_c(\sigma) + \log\rho^{\psi}_0(\sigma\mid c),
    \label{eq:dtUt}
  \end{equation}
  constant in $t$ at fixed $\sigma$, and every single-move difference entering the
  Metropolis factor $M_t$ is
\begin{equation}
  U_t(\sigma')-U_t(\sigma)
   = t\,\beta\big(\mathcal H_c(\sigma')-\mathcal H_c(\sigma)\big)
   - (1-t)\log\frac{\rho^{\psi}_0(\sigma'\mid c)}{\rho^{\psi}_0(\sigma\mid c)}
\end{equation}
  with $\sigma'=\mathrm{Swap}(\sigma,i,\tau)$. A single backward pass through
  Eq.~\eqref{eq:pinn} therefore returns gradients for $(\theta,\phi,\psi)$ at once.

\subsection{The prior objective: generalized VarGrad with a symmetry penalty}
\label{sec:vargrad}

Section~\ref{sec:related} contrasts FrOGS with the autoregressive samplers that
are trained on a reverse KL divergence, and \S\ref{sec:coupling} asserts that our
prior is not mode-seeking because it is not trained on that objective. This
appendix makes the assertion precise and states its limits. We show that the gradient FrOGS applies to its autoregressive prior is equivalent to VarGrad~\citep{richter2020vargrad} when the transport is switched off, and how a symmetry penalty keeps the prior from dropping symmetry-related modes when the transport is on.
Throughout we write the unnormalized target as
$\tilde p_c(\sigma) = e^{-\beta\mathcal{H}_c(\sigma)}$, so that
$p_c = \tilde p_c / Z(c)$, and $\psi$ as the weights of the
autoregressive prior $\rho^{\psi}_0$ as in \S\ref{sec:gradient}.

Given an unnormalized target $\tilde p$ and a normalized model $q_\psi$, the
log-variance loss \citep{richter2020vargrad, nusken2021logvar} is
\begin{equation}
  \mathcal{L}^{\mathrm{VG}}_{\mathcal{R}}(\psi)
  \;=\; \operatorname{Var}_{\sigma\sim\mathcal{R}}
        \left[\log\frac{q_\psi(\sigma)}{\tilde p(\sigma)}\right],
  \label{eq:vargrad}
\end{equation}
where $\mathcal{R}$ is any distribution whose support contains those of
$q_\psi$ and $\tilde p$. 

Importantly, $\mathcal{R}$ can be independent from $q_\psi$. The loss may therefore be evaluated on states the model itself would rarely produce, which a reverse KL cannot do because its expectation is taken under the model. In fact, VarGrad contains the reverse-KL training as a special case \citep{richter2020vargrad}: taking $\mathcal{R}=\mathrm{sg}[q_\psi]$ (stop-gradient) recovers the leave-one-out REINFORCE gradient of the reverse KL up to a factor. 

We now prove that FrOGS with the transport switched off is VarGrad.

\begin{proposition}
\label{prop:vargrad-reduction}
At a fixed condition $c$, set $Q^\theta_t\equiv 0$ and let the reference
measure be $\mathcal{R}_s\equiv\mathcal{R}$ for all $s$. Then the sample is
stationary, $X_t=X_0\sim\mathcal{R}$, and its log-weight $A_1$ \eqref{eq:weight} is
\begin{equation}
  A_1 \;=\; \log\frac{\tilde p_c(X_0)}{\rho^{\psi}_0(X_0\mid c)},
  \label{eq:A1-degenerate}
\end{equation}
and consequently
$\operatorname{Var}[A_1]=\mathcal{L}^{\mathrm{VG}}_{\mathcal{R}}(\psi)$.
Moreover, if the free-energy head attains its conditional optimum at every $s$,
then the PINN loss \eqref{eq:pinn} equals the same quantity,
\begin{equation}
  \min_{\phi}\ \mathcal{L}( \phi,\psi;1)
  \;=\; \mathcal{L}^{\mathrm{VG}}_{\mathcal{R}}(\psi),
  \qquad\text{attained at}\quad
  \partial_s F^{\phi}_s = \mathbb{E}_{\mathcal{R}}\!\left[\partial_s U_s\right],
  \label{eq:pinn-is-vargrad}
\end{equation}
so the bound $\mathcal{L}^{\mathrm{log\text{-}var}}\le t^2\mathcal{L}$ of
Eq.~\eqref{eq:pinn} holds with equality at $t=1$.
\end{proposition}

\begin{proof}
With $Q^\theta_t=0$, the sum in the RHS of Eq.~\eqref{eq:corrector} is 0, so 
\begin{align*}
    \mathcal{K}_s\rho_s(\sigma)&=-\partial_s U_s(\sigma) \\
    &=-\beta\mathcal{H}_c(\sigma)-\log\rho^{\psi}_0(\sigma\mid c),
\end{align*}
by Eq.~\eqref{eq:dtUt}. A sample with zero rates never moves, so $X_t=X_0 \ \forall t$. Eq.~\eqref{eq:weight} integrates a
constant along $t = 0 \rightarrow1$ and returns Eq.~\eqref{eq:A1-degenerate}:
\begin{align*}
    A_1&=\int_0^1 \mathcal{K}^\theta_s\rho_s(X_s)\,ds \\
  &=\int_0^1 \mathcal{K}^\theta_s\rho_s(X_0)\,ds\\ &=\int_0^1(-\beta\mathcal{H}_c(X_0)-\log\rho^{\psi}_0(X_0\mid c))\,ds \\
  &=-\beta\mathcal{H}_c(X_0)-\log\rho^{\psi}_0(X_0\mid c) \\
  &=\log\frac{\tilde p_c(X_0)}{\rho^{\psi}_0(X_0\mid c)}.
\end{align*}
For the second claim, the integrand of the LEAPS PINN loss \eqref{eq:pinn} at time
$s$ is $\mathbb{E}_{\mathcal{R}}\big[|{-\partial_s U_s}+\partial_s
F^{\phi}_s|^2\big]$; minimizing this over the scalar $\partial_s F^{\phi}_s$
gives the stated optimum and the residual
$\operatorname{Var}_{\mathcal{R}}[\partial_s U_s]$, which is $s$-independent and
therefore remains unchanged after the average over $s\sim\mathrm{Unif}_{[0,1]}$.
\end{proof}

Eq.~\eqref{eq:pinn-is-vargrad} identifies $\partial_s F^{\phi}_s$ as the term that cancels out the mean and leaves the variance behind. The estimator of \citet{richter2020vargrad}
centers by the minibatch mean, which is legitimate because every sample in the
batch shares one target. A FrOGS minibatch does not since it mixes times $s$ and conditions $c$, whose free energies differ by many $k_{\mathrm{B}}T$
across the training box. The learned $F^{\phi}_s(c)$ supplies the per-$(s,c)$
baseline that repairs this.

\paragraph{The general case.}
For $Q^\theta_t\neq0$, we aim to control the path-space
log-variance $\operatorname{Var}[A_1]$, which is
the path-measure divergence of \citet{nusken2021logvar} carried onto a discrete
state space \citep{RichterBerner2024, holderrieth2025leaps}. Two features
of the reduction nonetheless persist.

Writing
$r_s(\sigma)\equiv\tilde{\mathcal{K}}^\theta_s\rho_s(\sigma)+\partial_s F^{\phi}_s$ for
the pointwise residual of Eq.~\eqref{eq:pinn}, the gradient splits into a term
carried by the annealing path and a term carried by the Metropolis
acceptance:
\begin{align}
  \nabla_\psi\mathcal{L}
  &= 2\,\mathbb{E}_{s,\,\sigma\sim\mathcal{R}_s}
       \big[\,r_s(\sigma)\,\nabla_\psi r_s(\sigma)\,\big]
   \;=\; \mathcal{G}^{\mathrm{path}}+\mathcal{G}^{\mathrm{acc}},
  \label{eq:psi-grad}\\
  \mathcal{G}^{\mathrm{path}}
  &= -\,2\,\mathbb{E}_{s,\,\sigma\sim\mathcal{R}_s}
       \Big[\,r_s(\sigma)\ \nabla_\psi\log\rho^{\psi}_0(\sigma\mid c)\Big],
  \label{eq:psi-grad-path}\\
  \mathcal{G}^{\mathrm{acc}}
  &= 2\,\mathbb{E}_{s,\,\sigma\sim\mathcal{R}_s}
       \Big[\,r_s(\sigma)\!\!\sum_{i,\,\tau\neq\sigma_i}\!\!
         (1-s)\,M_s\,G_s^\theta\,\mathbf{1}\{\Delta_{i\tau}U_s>0\}\,
         \nabla_\psi\log\frac{\rho^{\psi}_0(\sigma'\mid c)}
                             {\rho^{\psi}_0(\sigma\mid c)}\Big].
  \label{eq:psi-grad-acc}
\end{align}

\begin{proof}[Derivation]
Differentiating $\mathcal{L}=\mathbb{E}[r_s^2]$ gives the first equality of
Eq.~\eqref{eq:psi-grad}, the reference $\mathcal{R}_s$ being held fixed and
$\partial_sF^{\phi}_s$ carrying no $\psi$, so that
$\nabla_\psi r_s=\nabla_\psi\tilde{\mathcal{K}}^\theta_s\rho_s$. By
Proposition~\ref{prop:reparam},
\begin{equation*}
  \tilde{\mathcal{K}}^\theta_s\rho_s(\sigma)
  = -\,\partial_s U_s(\sigma)
    +\sum_{i,\,\tau\neq\sigma_i}
      M_s(\tau,i\mid\sigma)\,G^{\theta}_s(\tau,i\mid\sigma),
\end{equation*}
and $G^{\theta}_s$ is a function of $(\sigma,s,c)$ alone. The prior therefore
enters in two places.
 
For $\mathcal{G}^{\mathrm{path}}$,

\begin{align*}
    U_t(\sigma\mid c) &= -(1-t)\log\rho^{\psi}_0(\sigma\mid c) + t\,U_1(\sigma\mid c) \ \text{by Eq.}~\eqref{eq:linpath}, \\
    \nabla_\psi[-\partial_sU_s(\sigma)]&=-\nabla_\psi\log\rho^{\psi}_0(\sigma\mid c).
\end{align*}
 
For $\mathcal{G}^{\mathrm{acc}}$, differentiate
$M_s=e^{-[\Delta_{i\tau}U_s]_+}$ away from the measure-zero kink at
$\Delta_{i\tau}U_s=0$:
\begin{align*}
  \nabla_\psi M_s
  &= -\,M_s\,\mathbf{1}\{\Delta_{i\tau}U_s>0\}\,\nabla_\psi\Delta_{i\tau}U_s, \\
  \qquad
  \nabla_\psi \Delta_{i\tau}U_s
  &= -(1-s)\,\nabla_\psi\log\frac{\rho^{\psi}_0(\sigma'\mid c)}
                                {\rho^{\psi}_0(\sigma\mid c)},
\end{align*}
the second line being the single-move difference of \S\ref{sec:gradient}.
Multiplying by $G^{\theta}_s$ and summing gives $\mathcal{G}^{\mathrm{acc}}$.
\end{proof}

We apply a stop gradient to the $\mathcal{G}^{\mathrm{acc}}$ term, so that only the $\mathcal{G}^{\mathrm{path}}$ term flows into the autoregressive network. Thus, our setup is the generalized case of VarGrad for the autoregressive network part.

% We now argue $\mathcal{G}^{\mathrm{acc}}$ is secondary to $\mathcal{G}^{\mathrm{path}}$. For the sigmoid
% parametrization of $\rho^{\psi}_0$, the derivative of the per-site
% $\log\rho^{\psi}_0$ with respect to the head's pre-activation output is
% $1-\rho^{\psi}_0\in(-1,1)$, bounded regardless of how small
% $\rho^{\psi}_0$ becomes. $G_s^\theta$ is a continuous function of $\theta$ and is also bounded. Thus, on a
% compact weight set $\Theta$, which the weights stay in after finitely many bounded updates, every factor of the $\mathcal{G}^{\mathrm{acc}}$ summand is
% bounded uniformly in $(\sigma,i,\tau)$. Three structural factors then suppress it
% relative to $\mathcal{G}^{\mathrm{path}}$: the indicator
% $\mathbf{1}\{\Delta_{i\tau}U_s>0\}$ discards downhill moves,
% $M_s=e^{-[\Delta_{i\tau}U_s]_+}<1$ damps the rest, and the prefactor $(1-s)$
% removes the term as $s\to1$, whereas $\mathcal{G}^{\mathrm{path}}$ carries no such
% prefactor and retains full weight along the path. We therefore focus on
% $\mathcal{G}^{\mathrm{path}}$.

$\mathcal{G}^{\mathrm{path}}$ raises $\log\rho^{\psi}_0(\sigma\mid c)$ on reference states
where $r_s>0$ and lowers it where $r_s<0$. Against a reference that covers a mode, abandoning
that mode is penalized: if $\rho^{\psi}_0$ leaves a set $\mathcal{A}$ empty on which
$\tilde p_c$ has mass while $\mathcal{R}_s$ still samples $\mathcal{A}$, then
$\ell=\log(\rho^{\psi}_0/\tilde p_c)\to-\infty$ there and inflates
$\operatorname{Var}_{\mathcal{R}_s}[\ell]$.

\paragraph{Symmetry penalty.}
The log-variance loss constrains $\rho^{\psi}_0$ on states from $\mathcal{R}_s$, which FrOGS draws from a replay buffer of training rollouts that start at the prior and are then moved by the transport and the injected MCMC steps. Thus, we want the prior to output states from all degenerate modes, so we add Eq.~\eqref{eq:sympen} to the loss. For any reference, the mean of the log-ratio in Eq.~\eqref{eq:sympen} is a difference of two KL divergences, so Jensen's inequality gives
\begin{equation}
    \mathcal{L}_{\mathrm{sym}}\;\ge\;\mathbb{E}_{c,\gamma}\Big[\Big(\mathrm{KL}\big(\mathcal{R}(\cdot\mid c)\,\big\|\,\rho^{\psi}_0(\gamma\,\cdot\mid\gamma c)\big)-\mathrm{KL}\big(\mathcal{R}(\cdot\mid c)\,\big\|\,\rho^{\psi}_0(\cdot\mid c)\big)\Big)^{2}\Big],
\end{equation}
where $\mathcal{R}(\cdot\mid c)$ is the distribution of buffer states at condition $c$, regardless of the time $s$. Both are forward KL divergences from the reference, which grow infinitely when $\rho^{\psi}_0$ loses a mode that the reference has. If the prior keeps a mode $\mathcal{A}$ of the reference but drops its image $\gamma\mathcal{A}$, the first diverges while the second stays finite, and so does the penalty. The dropped mode is thus penalized.

For CuAu, $\mathcal{S}$ maps the six $\text{L1}_0$ variants onto one another, and likewise the four $\text{L1}_2$ variants of each compound. For the Ising model at $h=0$, it maps the $m>0$ mode onto the $m<0$ mode, and vice versa. Every phase in Table~\ref{tab:chi2} therefore keeps all its variants as long as $\mathcal{R}_s$ samples one of them. Modes that no symmetry relates, such as two coexisting phases, rely on $\mathcal{R}_s$ sampling each of them.

Training starts from a near-uniform prior and turns the Hamiltonian on gradually. The rollouts initially end at $t=1/125$, and their end time is extended to $t=1$ over the first $6\times10^{4}$ of the $2\times10^{5}$ optimizer steps. $\mathcal{R}_s$ therefore starts out as broad as the uniform distribution, with no mode preferred, and the modes of the target emerge in it as the end time grows. All the modes surviving is an empirical result, which Table~\ref{tab:chi2} tests.

\subsection{Proof of the local equivariance of the U-net pyramid}
\label{sec:proof}

For the probability flux factorization~\eqref{eq:locequiv}, local equivariance follows if
the head is unchanged by the proposed single-site move:
\begin{equation}
  H_t^\theta(i\mid\mathrm{Swap}(\sigma,i,\tau))
   =H_t^\theta(i\mid\sigma).
  \label{eq:locinv}
\end{equation}
We prove this by tracking the physical sites that each feature can depend on
through pooling, convolution, and decoding.

Let $\Lambda$ denote the index set of the original lattice sites.
For each resolution level $\ell$ of the U-net pyramid, let $\Lambda_\ell$
denote the grid's index set. Level zero is the original grid, so
$\Lambda_0=\Lambda$, and each subsequent grid is obtained by pooling the previous one. A \emph{feature field} on $\Lambda_\ell$ is a collection
$f=(f_p)_{p\in\Lambda_\ell}$ of scalar or vector features, with
$f_p=f_p(\sigma,t,c)$ at each index $p$.

Starting from $x^{(0)}=\sigma$, non-overlapping average pooling with
per-axis strides $s^{(\ell)}\in\{1,2\}^d$ gives
\begin{equation}
  x_q^{(\ell+1)}
   =\frac{1}{|\pi_\ell^{-1}(q)|}
       \sum_{p\in\pi_\ell^{-1}(q)}x_p^{(\ell)},
  \qquad
  (\pi_\ell(p))_a=\lfloor p_a/s_a^{(\ell)}\rfloor,
  \label{eq:pool}
\end{equation}
where $\pi_\ell:\Lambda_\ell\to\Lambda_{\ell+1}$ maps each index to its
parent in the pooled grid, and $a=1,\dots,d$ indexes the coordinate axes.
Here $d=2$ for Ising and $d=3$ for the alloys. We include anisotropic
strides in the proof because the FrOGS sampler for CuAu contains a
$(1,1,2)$ stride for its $4\times4\times8$ site system.
Let $B_\ell(p)\subseteq\Lambda$ be the block of
original sites represented by index $p\in\Lambda_\ell$:
\begin{equation*}
  B_0(i)=\{i\},\qquad
  B_{\ell+1}(q)=\bigcup_{p\in\pi_\ell^{-1}(q)}B_\ell(p).
\end{equation*}
A level-$\ell$ field $f$ is \emph{cell-blind} if $f_p\text{ depends only on }
  \sigma_{\Lambda\setminus B_\ell(p)},\ t,\ c.$
At level zero, this means that the feature at site $i$ does not depend on
that site's occupation $\sigma_i$.

For the binary systems, the projector (which we inherit from LEAPS)
$P_t^\theta(\tau)=-\tfrac12(1+t)\tau w$ gives
\begin{equation}
  g_t(i\mid\sigma)
   \equiv G_t^\theta(-\sigma_i,i\mid\sigma)
   =(1+t)\sigma_iw^\top f_i^{(0)}(\sigma,t,c),
  \label{eq:binary-head}
\end{equation}
where $w\in\mathbb R^C$ is learned and $f^{(0)}$ is the finest-scale decoder
output. If $f^{(0)}$ is cell-blind, flipping $\sigma_i$ leaves
$f_i^{(0)}$ unchanged and negates $g_t(i\mid\sigma)$, as required.
For $K\ge2$, the readout
$(w_\tau-w_{\sigma_i})^\top f_i^{(0)}$ has the same property.
The cell-blind property of $f^{(0)}$ is the desired local invariance,
so proving the property is our goal.

\begin{lemma}[Pooling and upsampling]
\label{lem:partition}
The cells $B_\ell(p)$ partition $\Lambda$ and satisfy
$B_\ell(p)\subseteq B_{\ell+1}(\pi_\ell(p))$.
Nearest-neighbor upsampling to the original axis length returns each
index's own parent: for an axis of length $L$, stride $s\in\{1,2\}$, and
pooled length $m=\lceil L/s\rceil$,
\begin{equation}
  \lfloor nm/L\rfloor=\lfloor n/s\rfloor,
  \qquad n=0,\dots,L-1.
  \label{eq:upsample-parent}
\end{equation}
\end{lemma}

\begin{proof}
The nonempty sets $\pi^{-1}(k)=[ks,(k+1)s)\cap\{0,\dots,L-1\}$
partition each axis. Their products partition the grid, and the recursive
definition of $B_\ell$ gives nesting.
For $s=1$, Eq.~\eqref{eq:upsample-parent} is immediate.
For $s=2$, write $e=2m-L\in\{0,1\}$ and $n=2\lfloor n/2\rfloor+r$,
where $r\in\{0,1\}$. Then
\begin{equation*}
  \frac{nm}{L}=\lfloor n/2\rfloor+\frac r2+\frac{ne}{2L},
  \qquad 0\le\frac r2+\frac{ne}{2L}<1,
\end{equation*}
which proves the identity, including for odd $L$.
\end{proof}

\begin{lemma}[State-dependent convolution]
\label{lem:sdc}
Let $h$ be cell-blind at level $\ell$, and suppose every axis at that level
has length at least two. Under circular indexing, the center-masked convolution
\begin{equation}
  y_p=\sum_{\delta\in\{-1,0,1\}^d}
         k_\delta(h_p,t)x_{p\oplus\delta}^{(\ell)},
  \qquad k_0=0,
\end{equation}
is cell-blind.
\end{lemma}

\begin{proof}
The coefficients depend only on sites outside $B_\ell(p)$.
With axis lengths at least two, no nonzero offset wraps to
$p$, so every retained $x_{p\oplus\delta}^{(\ell)}$ also depends on a cell
disjoint from $B_\ell(p)$. Products and sums preserve this exclusion.
\end{proof}

The axis-length condition is critical. On a length-one axis, a noncentral
offset can wrap back to the excluded cell.

\begin{lemma}[Pointwise operations and global context]
\label{lem:pointwise}
Cell-blind fields remain cell-blind under pointwise activations, channel
mixing, concatenation, and FiLM modulation
$\gamma(c)\odot h_p+\beta(c)$. The leave-one-out context
\begin{equation*}
  \Gamma_p=
   \frac{\sum_{q\in\Lambda_\ell}x_q^{(\ell)}-x_p^{(\ell)}}{|\Lambda_\ell|-1}
\end{equation*}
is also cell-blind and may be concatenated with these fields.
\end{lemma}

\begin{proof}
Pointwise operations use only the input fields, while FiLM adds dependence
on $c$ alone. The sum defining $\Gamma_p$ excludes $B_\ell(p)$ because all
remaining cells are disjoint from it.
\end{proof}

Thus, local invariance permits global context as long as that context excludes
the site's own cell.

\begin{proposition}
\label{prop:locinv}
If every axis at the coarsest level has length at least two, all encoder,
upsampled, and decoder features are cell-blind. In particular, $f^{(0)}$
satisfies Eq.~\eqref{eq:locinv}.
\end{proposition}

\begin{proof}
At each level, the first encoder layer is a fixed center-masked convolution,
a special case of Lemma~\ref{lem:sdc} with constant coefficients.
Subsequent state-dependent convolutions, FiLM layers, activations, and
optional leave-one-out concatenations preserve cell-blindness by
Lemmas~\ref{lem:sdc} and~\ref{lem:pointwise}. Hence every encoder output
$u^{(\ell)}$ is cell-blind.

For the decoder, proceed from the coarsest level to the finest. If
$f^{(\ell+1)}$ is cell-blind, its value upsampled to $p$ comes from
$\pi_\ell(p)$ and excludes the entire parent cell
$B_{\ell+1}(\pi_\ell(p))$ by Lemma~\ref{lem:partition}.
It therefore excludes $B_\ell(p)$. Concatenation with the encoder skip,
$1\times1$ channel mixing, FiLM, and the refinement convolutions preserve
this property. Induction yields cell-blindness of $f^{(0)}$, and
$B_0(i)=\{i\}$ gives Eq.~\eqref{eq:locinv}.
\end{proof}

\paragraph{Numerical verification.}
We check the trained AgPd ($5^3$, two isotropic pooling levels) and CuAu
($4\times4\times8$, one anisotropic then one isotropic pooling level)
networks at every site (Table~\ref{tab:equivariance_checks}). We also check
a sweep over the
conditioning box at fixed $(\sigma,t)$ and a control with the center mask
removed while retaining the trained weights.
\begin{table}[t]
  \centering
  \caption{Numerical checks of local equivariance. Here
  $\sigma'=\mathrm{flip}_i\sigma$ and $\varepsilon_{\mathrm{fp32}}=1.19\times10^{-7}$.}
  \vspace{\baselineskip}
  \begin{tabular}{p{0.63\linewidth}p{0.25\linewidth}}
    \toprule
    Check & Observed residual \\
    \midrule
    Flux sign: $|g_t(i\mid\sigma)+g_t(i\mid\sigma')|$
      & $O(\varepsilon_{\mathrm{fp32}}|g|)$ \\
    Head invariance: change in $w^\top f_i^{(0)}$ after a flip
      & $\sim10^{-7}$ \\
    Relative flux residual across the conditioning box
      & $\sim10^{-7}$ \\
    Flux sign with the center mask removed
      & $\sim10^{-1}$ \\
    \bottomrule
  \end{tabular}
  \label{tab:equivariance_checks}
\end{table}

 \subsection{The FiLM setup}
 \label{sec:film}
 We first rescale $c$ per axis to a normalized $\tilde c\in[-1,1]^2$ by $\tilde c = 2(c-\ell)/(u-\ell)-1$, where $[\ell,u]$ is the training box (e.g.\ $\Delta\mu\in[-0.4,0.4]$ and $T\in[200,900]\,$K for AgPd). 
The two conditioning axes differ by three orders of magnitude, so without the normalization the temperature channel would dominate the first linear layer of the conditioning MLP at initialization. We follow the generic setup: a small MLP maps $\tilde c$ to a per-channel scale $\gamma(\tilde c)$ and shift $\beta(\tilde c)$, and each hidden feature map $h$ is modulated as $\gamma(\tilde c)\odot h + \beta(\tilde c)$. The MLP is zero-initialized so that $\gamma\equiv 1,\ \beta\equiv 0$ (the identity) at step $0$. 

\subsection{Re-indexing the fcc lattice onto an integer grid}
\label{app:reindex}

Both networks of \S\ref{sec:method} are three-dimensional convolutions, which act on a dense integer array. The sites of an fcc alloy sit on a non-orthogonal Bravais lattice, and embedding
them in the conventional cubic cell leaves half of the sites empty. We instead label each site by its coordinates in the primitive basis, which turns the fcc lattice into a dense $(i,j,k)$ grid (Fig.~\ref{fig:fcc-reindex}). 

Write the fcc primitive vectors in terms of the conventional lattice parameter
$a$,
\begin{equation}
  \mathbf{a}_1 = \tfrac{a}{2}(0,1,1), \qquad
  \mathbf{a}_2 = \tfrac{a}{2}(1,0,1), \qquad
  \mathbf{a}_3 = \tfrac{a}{2}(1,1,0),
  \label{eq:primitive-basis}
\end{equation}
and assign to the integer triple $(i,j,k) \in \Lambda$ the physical site
$\mathbf{r}(i,j,k) = i\,\mathbf{a}_1 + j\,\mathbf{a}_2 + k\,\mathbf{a}_3
= \tfrac{a}{2}(j{+}k,\; i{+}k,\; i{+}j)$.
The map is a bijection from $\mathbb{Z}^3$ onto the fcc lattice. Geometrically, it is a shear
that tilts the $60^\circ$ angles between the primitive vectors to
$90^\circ$, so the lattice becomes simple-cubic in index space while
$E(\sigma)$ \eqref{eq:ce} continues to be
evaluated on the physical geometry to preserve the original thermodynamics.

\begin{figure}
    \centering
    \includegraphics[width=0.9\linewidth]{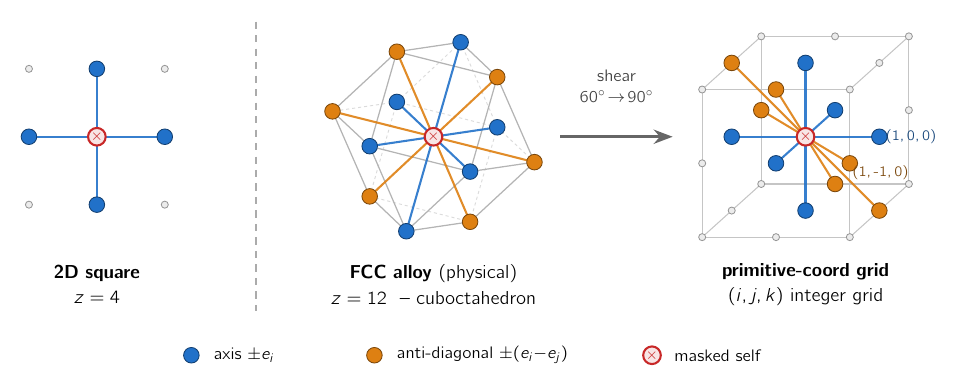}
    \caption{Diagram explaining the re-indexing of the fcc lattice onto an integer grid. The Ising model is on a 2D square lattice, so it needs no re-indexing (left). AgPd and CuAu are both 3D fcc lattices, so the shearing transformation can pack the crystal into a dense 3D simple cubic grid.}
    \label{fig:fcc-reindex}
\end{figure}

\subsection{Ablation study of prior-only and transport-only networks}
\label{sec:ablation}

Fig.~\ref{fig:ablation_ess} shows the ESS heatmaps of two ablated models, one with the transport set to identity ($Q^\theta_t = 0, \ \forall t$), and the other with the prior set to uniform following the original LEAPS setup. Both are trained under the same conditions as the full FrOGS model. Both models sample substantially less effectively than FrOGS (Fig.~\ref{fig:ess}). Notably, the ESS of the uniform prior model collapses $\rightarrow1/M$ in the ordered region, indicating the model cannot set up a transport to balance weights across samples. This ablation demonstrates that both components contribute substantially.

Fig.~\ref{fig:modecollapse-nl2enc1-uniform} shows that the raw output is balanced across degenerate modes, but the distribution after importance weighting concentrates on a single variant per phase (48--66\% of the weight, against a flat $1/g$ of 17--25\%). The variance in the importance weights explained in the first few paragraphs of \S\ref{sec:phasemix} causes this imbalance, since only a few samples have large importance weights. 

\begin{figure}
  \centering
  \begin{subfigure}[b]{0.48\linewidth}
    \includegraphics[width=\linewidth]{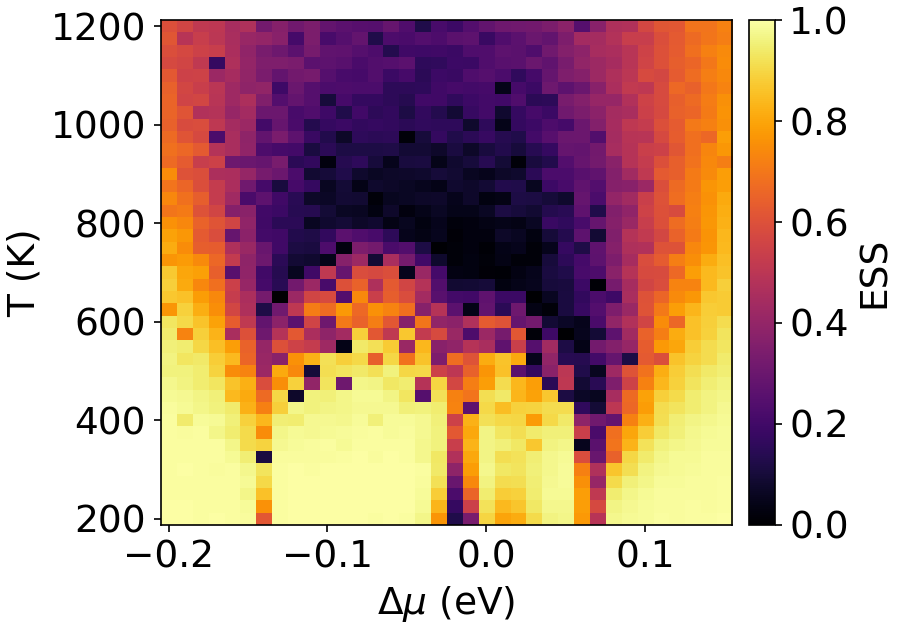}
    \caption{autoregressive only}
    \label{fig:auto_ess}
  \end{subfigure}
  \hfill
  \begin{subfigure}[b]{0.48\linewidth}
    \includegraphics[width=\linewidth]{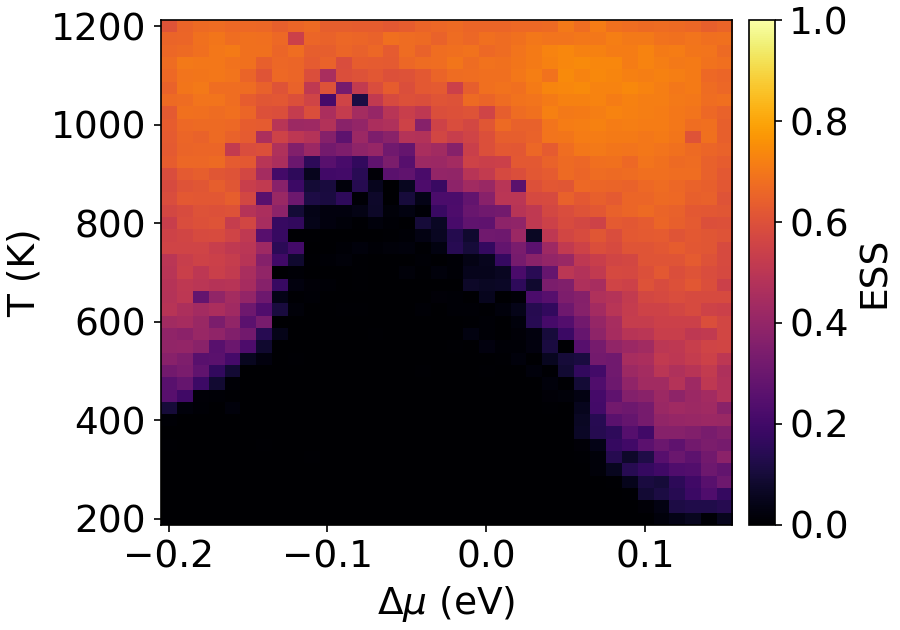}
    \caption{transport only}
    \label{fig:transport_ess}
  \end{subfigure}
  \caption{ESS plots of the ablated models.}
  \label{fig:ablation_ess}
\end{figure}

\begin{figure}[t]
\centering
\includegraphics[width=\textwidth]{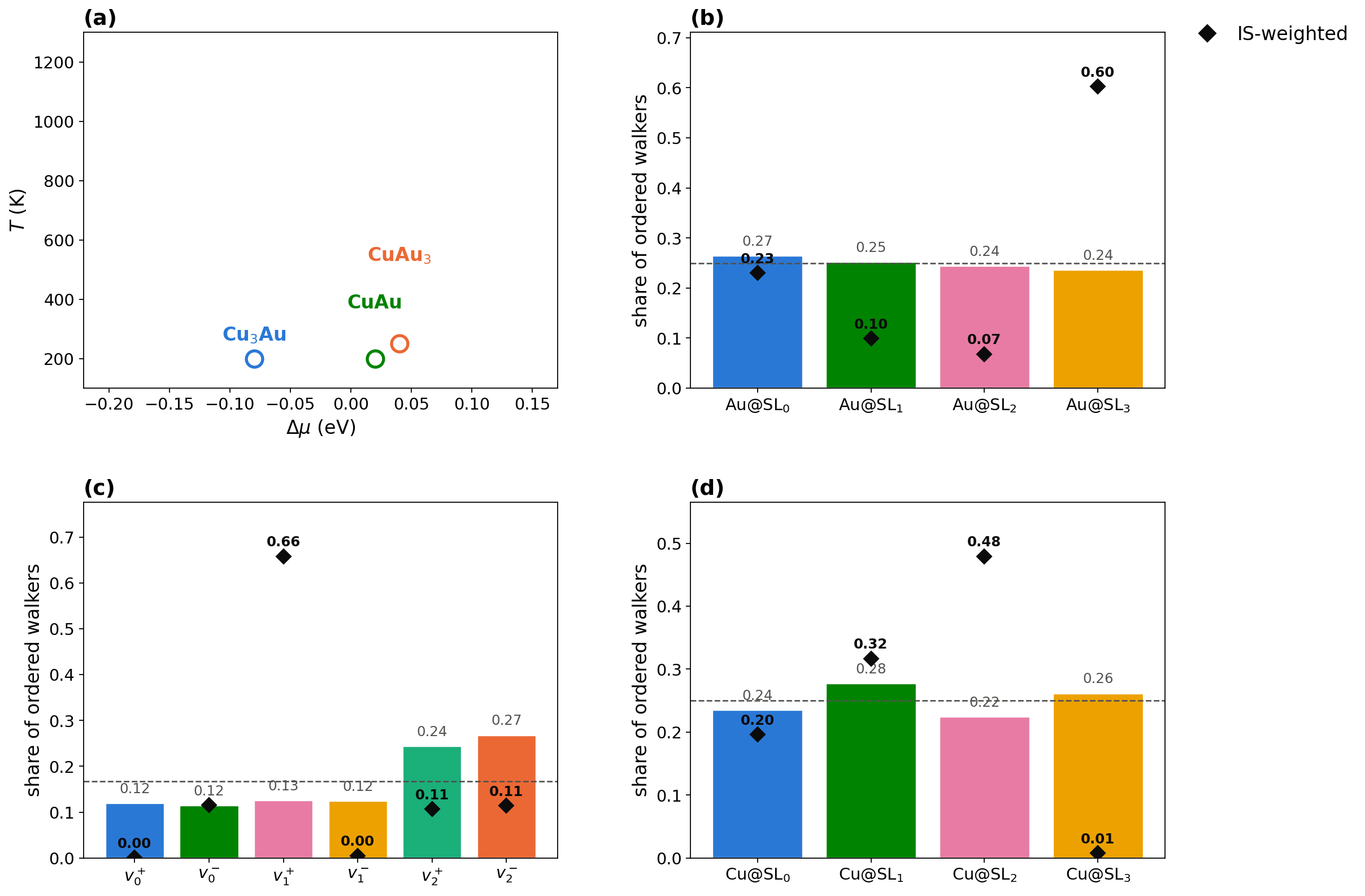}
\caption{Deepest-ordered-cell diagnostic for the uniform-prior CuAu ablation. \textbf{(a)}~Location in
$(\Delta\mu,T)$ of each phase's most-ordered grid cell (most gated samples, gating criterion explained in Table~\ref{tab:chi2}).
\textbf{(b--d)}~At each anchor, bars are the share of
unweighted ordered samples landing in each of the $g$ degenerate ordering variants of
that phase; black diamonds are the share after applying the $e^{A}$ importance weight. The bars sit close to
flat $1/g$ (dotted line), and all $g$ variants are populated everywhere, meaning
the raw sampler output does not collapse. However, the importance-weighted distribution concentrates on one variant at all three points.}
\label{fig:modecollapse-nl2enc1-uniform}
\end{figure}

 \subsection{Model and training details}
 \label{sec:training}

\begin{figure}
  \centering
  \includegraphics[width=\linewidth]{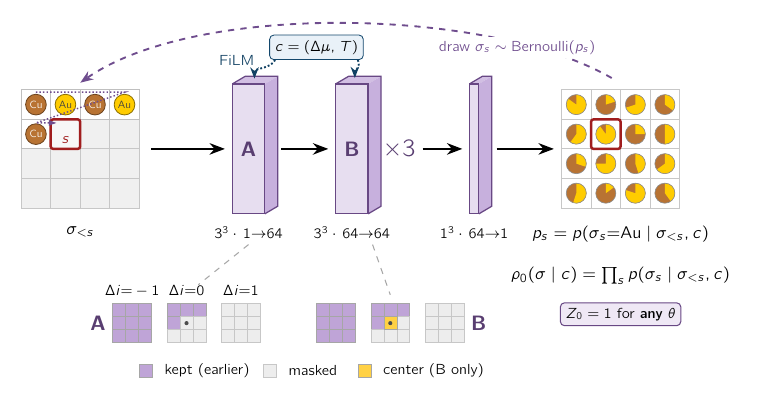}
  \caption{Architecture of the autoregressive prior. A type-A causally masked three-dimensional convolution (center excluded) followed by three type-B layers (center included) and a $1^3\ (64 \rightarrow1)$ head emits a conditional probability. We use zero padding, not circular, to conceal future sites.}
  \label{fig:auto}
\end{figure}

We first optimize the autoregressive part of our network by comparing the training-time statistics (training loss and corrected standard deviation of the log-weights) over the final 10k optimizer steps on the $4\times4\times8$ CuAu system, which has a long-range ordering that the autoregressive part should model (Table \ref{tab:cuau-l8-tail}). The transport part is fixed at $n_\ell=2$ (the number of decoder blocks), $n_\text{enc}=1$ (the number of convolutional layers in a block), and Q-net width 16. Both the training loss and the corrected importance weight standard deviation ($\mathrm{std}[A_{\mathrm{corr}}]$, defined subsequently) are logged once per optimizer step, so each entry summarizes $\sim10^4$ samples. $A_{\mathrm{corr}} = A_1 + F^{\phi}_{1}(c) - F^{\phi}_{0}(c)$ to be comparable across different $(\Delta\mu, T)$. train\_loss is a squared residual with a heavy right tail, so we report the median and the 99th percentile. $\mathrm{std}[A_{\mathrm{corr}}]$ is the per-step spread of the corrected work and is the quantity ESS responds to.

\begin{table}[t]
     \caption{Training time statistics for the $4\times4\times8$ CuAu model over the hyperparameters of FrOGS. The bolded row is the setup we choose.}
  \centering
  \vspace{\baselineskip}
  \begin{tabular}{llcccc}
    \toprule
    & & \multicolumn{2}{c}{train\_loss} & \multicolumn{2}{c}{$\mathrm{std}[A_{\mathrm{corr}}]$} \\
    \cmidrule(lr){3-4}\cmidrule(lr){5-6}
    AR width & $n_{\mathrm{mcmc}}$ & median & p99 & median & p99 \\
    \midrule
        16                     & 10 & 1.38 & 12.9 & 0.91 & 4.26 \\
    32                     & 10 & 0.79 & 5.96 & 0.60 & 4.45 \\
    64                     & 0 & 0.52 & 4.64 & 1.01 & 5.52 \\
    64                     & 10 & 0.58 & 4.02 & 0.52 & 2.26 \\
    \textbf{64}            & \textbf{40} & \textbf{0.53} & \textbf{4.19} & \textbf{0.41} & \textbf{2.09} \\
    64                     & 200 & 0.47 & 4.00 & 0.36 & 1.51 \\
    \bottomrule
  \end{tabular}
  \label{tab:cuau-l8-tail}
\end{table}

For the optimization of the hyperparameters of the transport part of the network, we compare the training-time statistics over the final 10k optimizer steps on the 2D-Ising model with $L=20$ (Table~\ref{tab:ablation}). The autoregressive part is fixed at width 64 and $n_{\mathrm{mcmc}}$ is set equal to 40. The flat + enc10 model had nonfinite loss for 75 of 10000 loss entries, which we excluded when computing the table entries. 

We choose the pyramid architecture with 3 encoders plus 2 decoders ($n_\ell=2$) and the autoregressive model width of 64 based on these experiments. We select $n_\text{mcmc}=40$ for training because $n_\text{mcmc}=200$ lowers the median $\mathrm{std}[A_{\mathrm{corr}}]$ by a further 0.05, less than the spread between training seeds (0.35--0.43 over four seeds), at five times the MCMC cost. All models are trained with the symmetry penalty, with $\lambda=0.1$. The train\_loss entries of Tables~\ref{tab:cuau-l8-tail} and~\ref{tab:ablation} exclude it. The autoregressive network is registered as a second Adam parameter group with a slower step size ($3\times10^{-4}$ for the transport and free-energy networks,
  $1\times10^{-4}$ for the prior, in every experiment reported here). We use $k=125$ discretization steps from $t=0$ to $1$.

\begin{table}[t]
    \caption{Training time statistics for the $L=20$ 2D Ising model over the hyperparameters of the transport part of FrOGS. The bolded row is the setup we choose.}
  \vspace{\baselineskip}
  \centering
  \begin{tabular}{llcccc}
    \toprule
    & & \multicolumn{2}{c}{train\_loss} & \multicolumn{2}{c}{$\mathrm{std}[A_{\mathrm{corr}}]$} \\
    \cmidrule(lr){3-4}\cmidrule(lr){5-6}
    Arm & $n_\text{enc}$ & median & p99 & median & p99 \\
    \midrule
        \textbf{pyramid} (\boldmath$n_\ell=2$) & \textbf{1} & \textbf{0.98} & \textbf{8.04} & \textbf{0.46} & \textbf{5.64} \\
                           &  2 & 1.06 & 12.5 & 0.47 & 20.8 \\
                           &  5 & 1.45 & 17.9 & 0.55 & 5.33 \\
    \midrule
    flat ($n_\ell=0$)      &  2 & 1.34 & 8.08 & 0.53 & 2.97 \\
                           &  5 & 1.47 & 8.47 & 0.54 & 3.78 \\
                           & 10 & 1.98 & 15.0 & 0.65 & 5.29 \\
                           & 20 & 1.80 & 9.64 & 0.60 & 3.70 \\
    \bottomrule
  \end{tabular}
  \label{tab:ablation}
\end{table}

\subsection{Multiscale 2D Ising result comparison}
\label{sec:multiL}

The learned sampler reproduces 2D-Ising thermodynamics at every cell size. The partition function tracks the Kaufman value to within $6 \times 10^{-4}$ per site across the full temperature range (Fig.~\ref{fig:multiL}b). The mass balance sits near 0.5 at all three sizes (Fig.~\ref{fig:multiL}c) due to the importance weights correcting the one-sided raw counts, and $G(r)$ at criticality follows the same-size Wolff reference mostly within $2\sigma$ for all scales (Fig.~\ref{fig:multiL}d). Panel (a) should not be interpreted as a quality ranking since the effective sample size at fixed M falls with N for any sampler, the log-weight variance being extensive. Dividing that out gives $\text{Var}[A]/N =0.024$, 0.028 and 0.019 at $T_c$, which is not monotonic in L.

\begin{figure}
    \centering
    \includegraphics[width=0.8\linewidth]{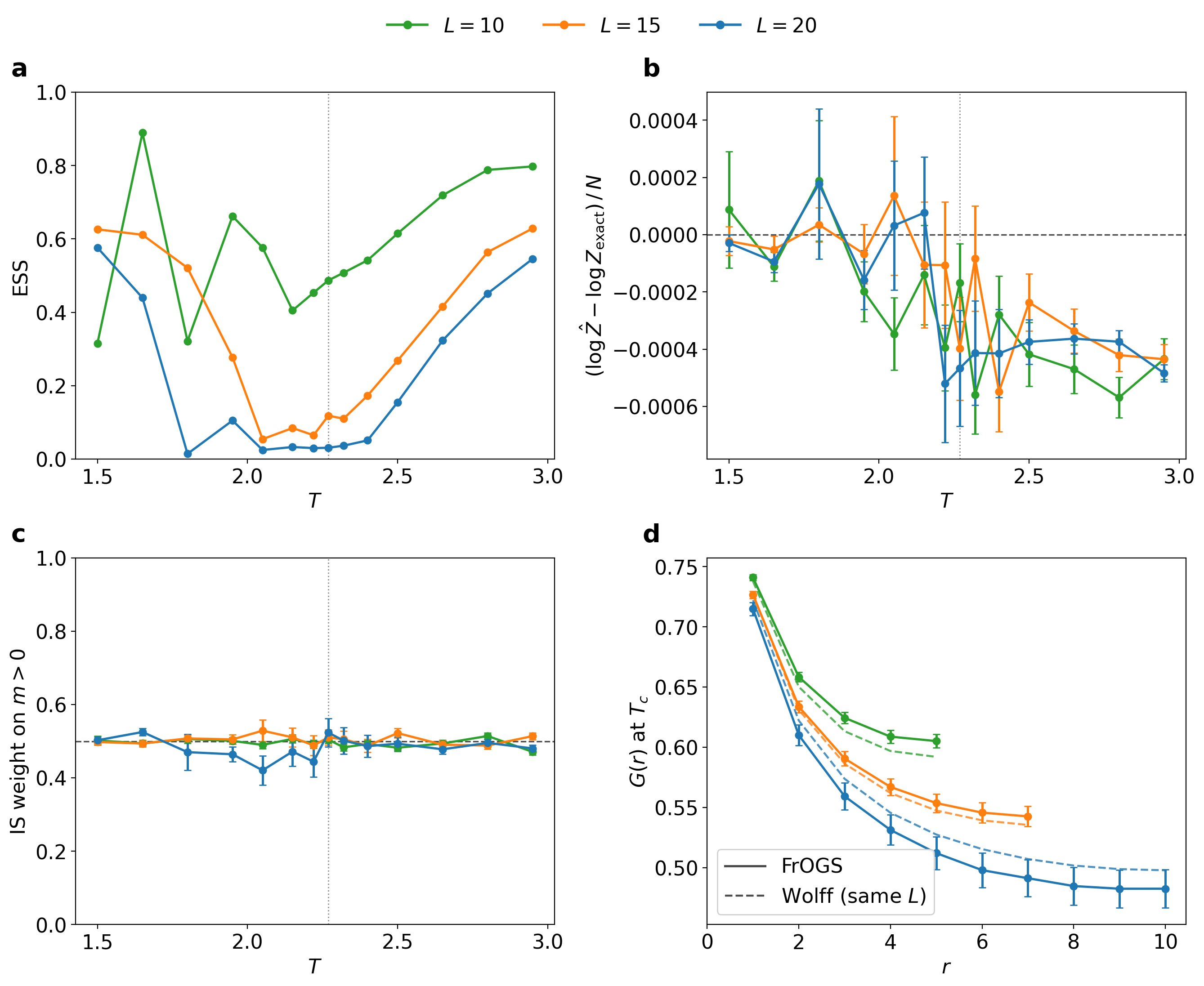}
    \caption{2D Ising at $L = 10, 15, 20$. FrOGS models with the same hyperparameters trained separately at each lattice size and evaluated on the h = 0 line. (a) ESS vs. $T$. (b) Partition-function error per site against the exact finite-$L$ Kaufman value. (c) Importance-weighted mass on $m > 0$, which is compared to the 0.5 reference (dotted line). (d) Correlation function at $T_c$ against a Wolff reference on the same lattice (dashed). $M = 5000$ walkers per temperature and no MCMC mixed in during evaluation.}
    \label{fig:multiL}
\end{figure}

\subsection{Sensitivity analysis of FrOGS phase count}
\label{sec:sensitivity}

We conduct a sensitivity analysis of FrOGS with the baseline discrete neural sampler SEGAL~\citep{damewood2022sampling}, showing our model discovers the CuAu$_3$ phase robustly regardless of the tie-line construction threshold of our phase diagram constructor (\S\ref{sec:phasediagram}). The phase diagram constructor certifies a convex-hull gap of $f(x)$ as a tie-line if it fulfills three conditions: a width of at least
four lattice steps, a chord depth exceeding three standard errors (or a measured absence), and an $x$-overlap with a certified gap at one temperature below, anchored at the lowest temperature. Thus, the current analysis has two threshold levers of $w_{\min}=4$ and $n_\sigma=3$, plus two additional rules of the measured absence and persistence (having an $x$-overlap with the row below). 

We count the number of the phases of the samples generated by both FrOGS and SEGAL~\citep{damewood2022sampling}, while sweeping across
  $(w_{\min}, n_\sigma)\in\{2,\dots,10\}\times \{1,\dots,6\}$, three different
  persistence rules, the measured-absence route on or off, and the bin floor for the $\chi^2$ statistic $W\in\{0,2,3,5,8,10,20\}$ of \S\ref{sec:phasediagram}.
  This results in $2268$ settings per sampler. We count phases as the number of compositions at which one tie-line ends and the next begins. A composition is admitted if it chains unbroken from the lowest temperature row, drifting by at most two lattice steps per row. Under this counting scheme, FrOGS resolves three ordered compounds and SEGAL two across all
settings, the difference being at
$x_\text{Au}=\tfrac34$. 

Additionally, we use four different seeds for training FrOGS models for the CuAu system and report the phase diagrams from each sampler in Fig.~\ref{fig:seeds}. This demonstrates that the phase diagram drawn with FrOGS agreeing with metadynamics was not a single lucky event. All four seeds resolve the three ordered compounds at $x_{\mathrm{Au}}=\tfrac14,\tfrac12,\tfrac34$ and match the metadynamics phase boundaries closely.

\begin{figure}
    \centering
    \includegraphics[width=\linewidth]{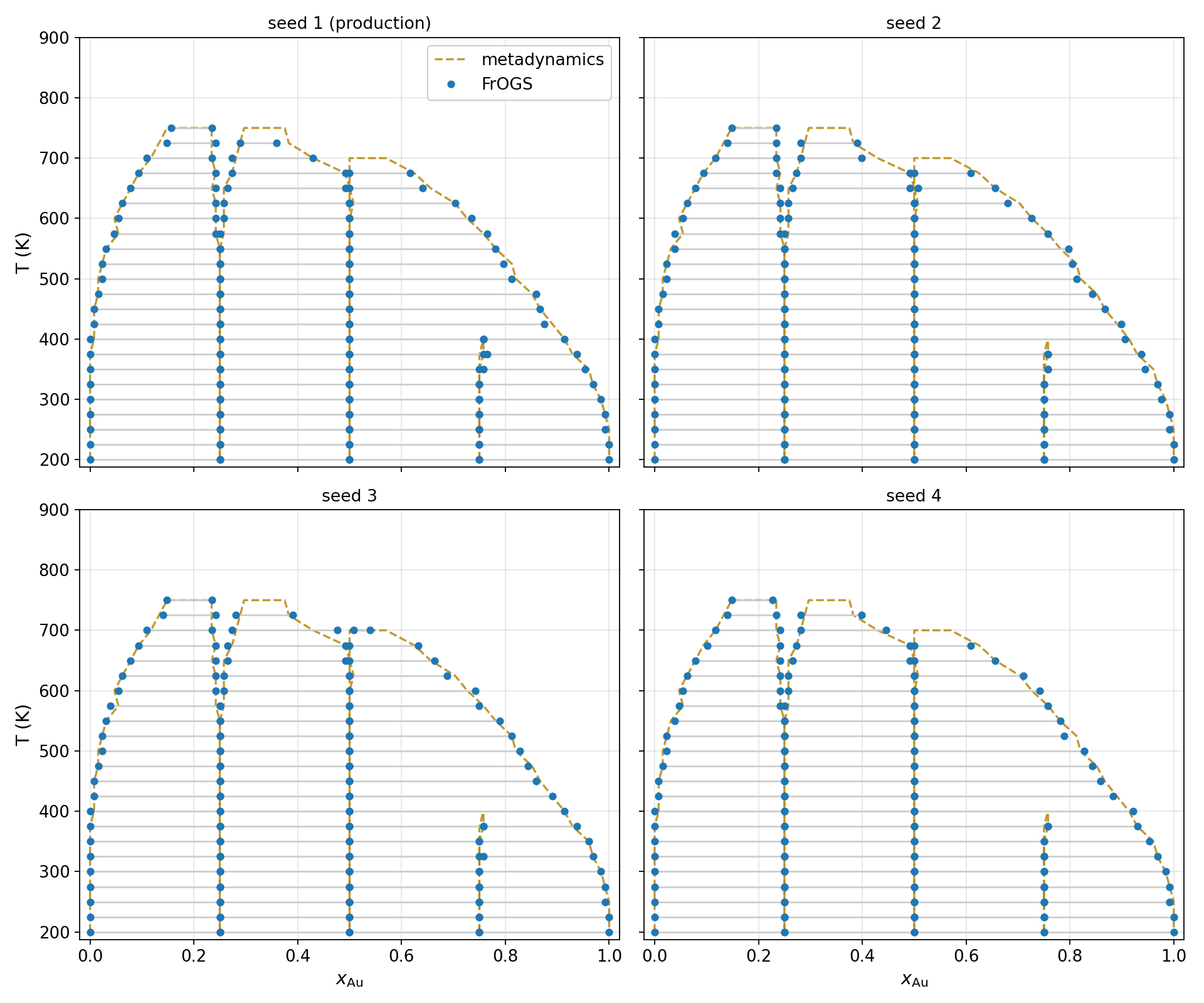}
    \caption{Phase diagram of 128-site CuAu, drawn with FrOGS initialized with four different seeds.}
    \label{fig:seeds}
\end{figure}

\subsection{FrOGS vs. conventional Monte Carlo}
  \label{sec:compute}

    \paragraph{FrOGS.} The CuAu model of \S\ref{sec:cuau} trained for $2\times10^{5}$ steps in $5.0$ hours on one NVIDIA A100. The sweep over $(\Delta\mu, T)$ took $53$ s per condition for $M=5000$ walkers over a $125$-step CTMC on the same
  GPU, equivalent to $22$ GPU-hours for the
  $36\times41$ grid.

  \paragraph{Conventional Monte Carlo.} We implemented a single-site semi-grand canonical
  Metropolis algorithm with local energy updates on the same CuAu model. It sustains $3.2\times10^{6}$ attempted flips per second on one CPU
  core, and that rate is flat to $\pm15\%$ over the whole $(\Delta\mu,T)$ box. 
  
The integrated autocorrelation time is around $N=128$ attempted flips in the disordered region and at least $5\times10^{6}$ attempted flips inside the low-temperature two-phase region. At each condition, measuring the number of passes it needs for a common $10^{-3}$ standard error on $x$, rows at $T\geq575$ K cost $0.32$ core-hours over $72$ conditions, while four conditions at $T\leq450$ K cost $15$ core-hours between them. If we go lower, the system cannot equilibrate. At $200$ K, seven of $36$ chemical potentials give chains that are stuck in different basins, with $x$ differing between seeds by as much as $0.14$ while each chain reports its own error bar as below $10^{-5}$. This is a classic case of a small error bar on a wrong answer. Anchoring one chain per phase, as it is commonly done in real workflows, removes this problem and lets the grid be computable in a few core-hours, at the price of supplying the phases in advance; FrOGS needs no such input.

\end{document}